\newif\iflongpaper
\longpapertrue

\documentclass[a4paper,USenglish,cleveref,autoref,thm-restate]{lipics-v2021}
\usepackage[utf8]{inputenc}
\usepackage{stmaryrd,amssymb,amsmath,amsthm,amsfonts}
\usepackage{algpseudocode}
\usepackage{xspace,xpunctuate}
\usepackage{hyperref}
\usepackage{cleveref}
\usepackage{nicefrac}
\usepackage{mathtools}
\usepackage{comment}
\usepackage{graphicx}
\usepackage[normalem]{ulem}
\usepackage{csquotes}
\usepackage[ruled,longend,vlined,linesnumbered]{algorithm2e}

\nolinenumbers

\pdfoutput=1 
\hideLIPIcs  

\newbox\brbox
\setbox\brbox\hbox{$[$}
\newbox\brsmallbox
\setbox\brsmallbox\hbox{$\scriptstyle[$}
\let\zupp[
\mathcode`\["8000
\def\crazy{\futurelet\next\dostuff}
\def\dostuff{\ifx\next\zupp\schnapp\else
  \mathchoice{\mathopen{\copy\brbox}}%
             {\mathopen{\copy\brbox}}%
             {\mathopen{\copy\brsmallbox}}%
             {\mathopen{\copy\brsmallbox}}\fi}
{\catcode`\[\active\gdef[{\crazy}
\gdef\schnapp#1]]{{\def[{\llbracket}#1\rrbracket}}
}

\let\oldsim=\sim
\def\sim{\mathbin{\oldsim}}

\newtheorem{fact}{Fact}
\theoremstyle{definition}

\newcommand{\N}{\ensuremath{{\mathbf{N}}}\xspace} 
\newcommand{\Z}{\ensuremath{{\mathbf{Z}}}\xspace} 
\newcommand{\X}{\ensuremath{{\mathcal{X}}}\xspace} %
\newcommand{\C}{\ensuremath{{\mathcal{C}}}\xspace}
\def\cal{\mathcal}
\newcommand{\Ww}{$W[1]$}
\newcommand{\FPT}{\textnormal{FPT}\xspace} 
\def\WReach{{\rm WReach}}
\DeclareMathOperator{\td}{td}

\newcommand{\bu}{{\bar u \xspace}}
\newcommand{\bv}{{\bar v \xspace}}
\newcommand{\bw}{{\bar w \xspace}}
\newcommand{\bx}{{\bar x \xspace}}
\newcommand{\by}{{\bar y \xspace}}
\newcommand{\bz}{{\bar z \xspace}}

\newcommand{\calC}{\mathcal{C}}
\newcommand{\vG}{{\smash{\vec{G}}}}
\newcommand{\vH}{{\smash{\vec{H}}}}

\newcommand{\fG}{\vec{G}_{\pi}}
\newcommand{\hG}{\vec{G}_{\pi}^2}

\renewcommand{\phi}{\varphi}
\renewcommand{\epsilon}{\varepsilon}

\def\col{\textnormal{col}}
\def\wcol{\textnormal{wcol}}
\def\dist{\textnormal{dist}}
\DeclareMathOperator{\wdist}{wdist}
\newcommand{\dom}{\textnormal{dom}}

\newcommand{\cnt}[1]{\#{#1}\,}

\newcommand{\FOC}{{\normalfont{FOC}}$(\mathbf{P})$\xspace}
\newcommand{\FOCless}{{\normalfont{FOC}}$(\{>\})$\xspace}
\newcommand{\FOCONE}{{\normalfont{FOC}}$_1(\mathbf{P})$\xspace}
\newcommand{\FOX}{{\normalfont{FO}$(\{{>}\kern1pt0\})$}\xspace}

\def\FO{{\rm FO}\xspace}

\newcommand{\suc}{\textnormal{succ}}
\newcommand{\enc}{\textnormal{enc}}
\newcommand{\enccomp}{\overline{\textnormal{enc}}}

\usepackage[textsize=tiny]{todonotes} 

\usepackage{framed}
\usepackage{mdframed}
\usepackage{float}
\usepackage{xcolor}

\newlength{\leftbarwidth}
\newlength{\leftbarsep}
\renewenvironment{leftbar}[1][blue]
{%
    \def\FrameCommand
    {%
        {\hspace{-8pt} \color{#1} \vrule width 3pt}%
        \hspace{0pt}
        \fboxsep=\FrameSep\colorbox{#1!5!}%
    }%
    \MakeFramed{\hsize\hsize\advance\hsize-\width\FrameRestore}%
}
{\endMakeFramed}
\crefname{algocf}{alg.}{algs.}
\Crefname{algocf}{Algorithm}{Algorithms}

\titlerunning{Restricted FO Counting Properties on Nowhere Dense Classes and Beyond}
\title{Evaluating Restricted First-Order Counting Properties on Nowhere Dense Classes and Beyond}

\titlerunning{Evaluating Restricted FO-Counting Properties on Nowhere Dense Classes and Beyond}

\author{Jan Dreier}{TU Wien}{dreier@ac.tuwien.ac.at}{https://orcid.org/0000-0002-2662-5303}{}

\author{Daniel Mock}{RWTH Aachen University}{mock@cs.rwth-aachen.de}{https://orcid.org/0000-0002-0011-6754}{Supported by the German Science Foundation DFG, grant no.\ DFG-927/15-2}

\author{Peter Rossmanith}{RWTH Aachen University}{rossmani@cs.rwth-aachen.de}{https://orcid.org/0000-0003-0177-8028}{Supported by the German Science Foundation DFG, grant no.\ DFG-927/15-2}

\authorrunning{J. Dreier, D. Mock, and P. Rossmanith}
\Copyright{Jan Dreier, Daniel Mock, and Peter Rossmanith}

\relatedversiondetails{Full version}{}
 


\EventEditors{Inge Li G{\o}rtz, Martin Farach-Colton, Simon J. Puglisi, and Grzegorz Herman}
\EventNoEds{4}
\EventLongTitle{31st Annual European Symposium on Algorithms (ESA 2023)}
\EventShortTitle{ESA 2023}
\EventAcronym{ESA}
\EventYear{2023}
\EventDate{September 4--6, 2023}
\EventLocation{Amsterdam, the Netherlands}
\EventLogo{}
\SeriesVolume{274}
\ArticleNo{62}

\begin{document}
\maketitle
\begin{abstract}
It is known that first-order logic with some counting extensions can
be efficiently evaluated on graph classes with bounded expansion, where
depth-$r$ minors have constant density.  More precisely, the formulas
are $\exists x_1\ldots x_k \#y\,\phi(x_1,\ldots,x_k, y)>N$,
where $\phi$ is an FO-formula.
If $\phi$ is quantifier-free, we can
extend this result to \emph{nowhere dense} graph classes with an
almost linear FPT run time.  Lifting this
result further to slightly more general graph classes, namely almost nowhere
dense classes, where the size of depth-$r$ clique minors is subpolynomial,
is impossible unless $\rm FPT=W[1]$.  On the other hand, in almost
nowhere dense classes we can approximate such counting formulas with a
small additive error.
Note those counting formulas are contained in \FOCless but not \FOCONE.

In particular, it follows that partial covering problems, such as partial
dominating set, have fixed parameter algorithms on nowhere dense
graph classes with almost linear running time.
\end{abstract}

\section{Introduction}

First-order logic can be used to express algorithmic problems.
FO-model checking on certain classes of structures is therefore a
meta-algorithm, which solves many problems at the same time.
For example, the three classical problems that started the research on
parameterized complexity are all FO-expressible:
Vertex Cover, Independent Set, and
Dominating Set~\cite{DowneyF1999,DowneyF2013}.  Dominating Set with
the natural parameter---the size of the minimal dominating set---is
$\rm W[2]$-complete on general graphs, but fixed parameter tractable (fpt)
on many special graph classes.
The study of \emph{sparsity}, initiated by Nešetřil and Ossona
de Mendez, has led to the concept of \emph{bounded expansion} and
\emph{nowhere dense} graph classes~\cite{sparsity}.
They generalize many well-known notions of sparsity, such as bounded degree,
planarity, bounded genus, bounded treewidth, (topological)
minor-closed, etc.\ and
have led to quite general algorithmic
results~\cite{Seese1996,FrickG2001,DawarGK2007,DvorakKT2013}.
Most notably, Grohe, Kreutzer, and Siebertz showed that FO-model checking
is fpt on nowhere dense graph classes~\cite{GroheKS17}.  This shows,
e.g., that dominating set is fpt on nowhere dense graphs, a result that
was already known:  Dawar and Kreutzer were able to find a specific
algorithm several years earlier~\cite{DawarK2009} that solves 
generalizations of the dominating set problem.  All of them are
FO-expressible, which shows how strong meta-algorithms are.

Partial dominating set, also called $t$-dominating set, is another
generalization of dominating set:  The input is a graph $G$ and two numbers
$k$ and~$t$.  The question is, whether $G$ contains $k$ vertices that
dominate at least $t$ vertices.  The parameter is $k$, as in the classical
dominating set problem.  (If you choose $t$ as the parameter---which also
makes sense---the problem becomes fixed-parameter tractable even on general
graphs~\cite{KneisMR2007}.)  
The length of an FO-formula expressing the existence of a partial dominating depends on $t$, which is not bounded by any function of $k$
and therefore all the results on first-order
model checking do not help when we parameterize by $k$ only.  Golovach and Villanger showed that partial
dominating set remains hard on degenerate graphs~\cite{GolovachV2008},
while Amini, Fomin, and Saurabh have shown that partial dominating set is
fixed-parameter tractable in minor-closed graph classes, which generalized
earlier positive results~\cite{AminiFS2011}.
Very recently, this was improved to graph
classes with bounded-expansion, while simultaneously using only linear
fpt time instead of polynomial fpt time, i.e, the running time is now
only $f(k)n$~\cite{DreierR2021}.

This result was achieved by another meta-theorem for the counting
logic \FOCless on classes of bounded expansion.
\FOCless is a fragment of the logic \FOC, introduced by Kuske and Schweikardt
in order to generalize first-order logic to counting problems~\cite{KuskeS2017}.
\FOC is a very expressive counting logic and allows counting
quantifiers $\#\by\phi(\bx,\by)$, which count for how many $\by$ the
\FOC-formula $\phi(\bx,\by)$ is true.  Moreover, arithmetic operations
are allowed as well as all predicates in~$\bf P$, which might contain
comparisons, equivalence modula a number, etc.  Kuske and Schweikardt
showed that the \FOC-model checking problem is fixed parameter
tracktable on graphs of bounded degree and hard on trees of bounded
depth.  The fragment \FOCless is more restrictive and allows only
counting quantifiers of single variables and no arithmetic operations.
The only predicate is comparison against an arbitrary number, but not
between counting terms.
While \FOCless-model checking is still hard on trees of
bounded depth, there is an ``approximation scheme'' for \FOCless on classes of
bounded expansion \cite{DreierR2021}:  An algorithm gives either the
right answer or says ``mayby,'' but only if the formula is both almost
satisfied and not satisfied.  For a fragment of \FOCless, which
captures in particular the partial dominating set problem, we can
compute even an exact answer to the model checking problem
in linear fpt time~\cite{DreierR2021}.
That fragment consists of formulas of the form   
\begin{equation}\label{equ:pdslike}
    \exists x_1 \dots \exists x_k \cnt y \phi(y, x_1, \dots, x_k) > N ,
\end{equation} 
where $\phi$ is a first-order formula and $N$ an arbitrary number. The
semantics of the \emph{counting quantifier} $\cnt y \phi(y, v_1, \dots,
v_k)$ is the number of vertices $u$ in $G$ such that $G$ satisfies
$\phi(u, v_1, \dots, v_k)$. As an example, the existence of partial
dominating set can be expressed as
\begin{equation}\label{formula:pds}
    \exists x_1 \dots \exists x_k \cnt y \bigvee_{i =1}^k E(y, x_i)
    \lor y = x_i > t,
\end{equation}
where $k$ is the number of the dominating, and $t$ the number of dominated
vertices. The length of the formula only depends on~$k$.  This implies
that partial dominating set can be solved in linear fpt time on classes of
bounded expansion. 

There is another fragment of \FOC, which should not be confused with
\FOCless. In \FOCONE, introduced by Grohe and Schweikardt \cite{GroheS18},
the counting terms may contain at most one free variable. They show
that \FOCONE is fixed-parameter tractable on nowhere dense graph classes
\cite{GroheS18}. Note that formula~\ref{formula:pds} is in \FOCless but
not in \FOCONE as the counting term relies on $k$ free variables. Hence,
\FOCless and \FOCONE are orthogonal in there expressiveness.

There has been some research about \emph{low degree graphs}. A graph class
has low degree if every (sufficiently large) graph has degree at most
$n^\varepsilon$ for every $\epsilon>0$. Examples are classes with bounded
degree or classes with degree bounded by a polylogarithmic function. These
graph classes are incomparable to nowhere dense classes. Especially,
classes of low degree are not closed under subgraphs. On those classes,
Grohe has shown that first-order model-checking can be solved in
almost linear time \cite{Grohe01}.  Recently, Durand, Schweikardt, and
Segoufin have generalized the result to query counting with constant
delay and almost linear preprocessing time~\cite{DurandSS22}. Vigny
explores dynamic query evaluation on graph classes with low degree~\cite{vigny}.

\emph{Almost nowhere dense} is a property which subsumes both low degree
and nowhere dense classes.  Whereas a nowhere dense class $\cal C$
can be characterized that for every $r$ graphs do not contain up to $r$
times subdivided cliques of arbitrary sizes, for an almost nowhere dense
class arbitrary sizes are allowed, but their growth must be bounded by
subpolynomial function in the size of the graph.

\begin{table}
\halign to \hsize{\strut#\hfil\tabskip=1em plus 2em&%
        #\hfil&%
	#\hfil&%
	#\hfil&%
	#\hfil\tabskip=0em\cr
\noalign{\hrule\smallskip}
Graph class&\hfil FO-MC&\hfil \FOCONE&\hfil \FOCless&\hfil PDS like\cr
\noalign{\smallskip\hrule\smallskip}
bounded expansion&fpt \cite{DvorakKT2013}&fpt \cite{GroheS18}&
                  hard \cite{DreierR2021}&
		  fpt \cite{DreierR2021}\cr
                  &&&
                  $(1+\epsilon)$-approx fpt \cite{DreierR2021}&
		  \cr
\noalign{\smallskip}
nowhere dense&fpt \cite{GroheKS17}&fpt \cite{GroheS18}&
                  hard, approx open &\bf fpt\rlap{$^c$}\cr
\noalign{\smallskip}
almost nowhere dense&\bf hard\rlap{$^a$}&
                  \bf hard\rlap{$^a$}&\bf hard\rlap{$^a$}&\bf hard\rlap{$^a$}\cr
                    &&
                  &&\bf approx$\pm\delta$ fpt$^b$\cr
\noalign{\smallskip}
general graphs&hard&hard&hard&hard\cr
\noalign{\smallskip\hrule\smallskip}
}
$^a$ \Cref{cor:pds-hard-awnd},
$^b$ \Cref{cor:approx-awnd},
$^c$ \Cref{thm:algo}

\smallskip
\caption{Results of this paper (in boldface) and some related known results.
\emph{Hard} means at least $\rm W[1]$-hard.
\emph{PDS like} indicates problems similar to the partial dominating
set problems:  All problems that can be expressed by a \FOCless
formula of the form~(\ref{equ:pdslike}).  The mentioned approximation
results are quite different.
Numbers are approximated either with a relative or an absolute
error.}
\label{tbl:results}
\end{table}


\subsection{Our Results}
In this work, we consider a fragment of \FOCless, which we will call \emph{PDS-like formulas}, namely formulas of the form
\[
\exists x_1 \dots \exists x_k \cnt y \phi(y, x_1, \dots, x_k) > N
\]
for a quantifier-free FO-formula $\phi$ and an (arbitrarily big) number $N
\in \Z$.  This logic is strong enough to express the partial dominating
set problem as formula (\ref{formula:pds}) is contained in the fragment
described above.  Remember that this fragment and \FOCONE are orthogonal.
Table~\ref{tbl:results} contains an overview of most of the results in
this paper.

In formulas that start with existential quantifier it is natural to
ask for a witness, if we can indeed fulfill the formula.  For example,
in the partial dominating set problem we are usually interested in
actually finding the dominating set rather than verify than one
exists.  Often, this is not an issue as problems are self-reducible.
Using self-reducibility to find a witness incurs a runtime penalty.
The next theorem shows that solving the model checking problem, and
finding a witness, for formulas in the form of~\ref{equ:pdslike} is
possible.

\begin{restatable}{theorem}{theoremalgo}
\label{thm:algo}
    Let $C$ be a nowhere dense graph class.
    For every $\epsilon >0$, every graph $G \in \cal C$ and every quantifier-free first-order formula $\phi(y\bx)$ we can compute a vertex tuple $\bu^*$ that maximizes $[[\cnt y \phi(y\bu^*)]]^G$ in time $O(n^{1+\epsilon})$.
\end{restatable}

As an immediate corollary, we get that the model-checking problem for PDS-like formulas and thus, also the partial dominating set problem
are solvable in almost linear fpt time on nowhere dense graph classes,
where the parameter is the length of the formulas or the solution size $k$ respectively.  Moreover, our meta-algorithm
does not only work for partial dominating set, but for variants such as
partial total or partial connected dominating set as well.

Note that \Cref{thm:algo} does not follow from the fact that model-checking for \FOCONE or that query-counting for FO-logic is fixed-parameter tractable \cite{GroheS18} as we do not count the number of solutions to a query, but the number of witnesses to some solution. Also, PDS-like formulas form a fragment orthogonal to \FOCONE. Moreover, we were not able to prove \Cref{thm:algo} by using the result from \cite{GroheS18} as a subroutine: formulas inside a counting quantifier are allowed to have at most one free variable and this weakens self-reducibility or similar techniques drastically.

The above theorem cannot be extended to the more general case of almost nowhere
dense graph classes. It turns out that even for non-counting formulas this is
not possible, as the (classical) dominating set problem becomes $\rm W[1]$-hard
on some almost nowhere dense graph classes. This lower bound implies as a
special case that plain FO-model checking is intractable on some almost nowhere
dense graph classes.  As far as we are aware this does not follow directly from
previously known results.

However, we can go beyond nowhere dense classes if we do not insist on an exact
solution: The model-checking problem for PDS-like formulas can be approximated with an
\emph{additive} subpolynomial error in almost linear fpt time on almost nowhere
dense classes of graphs. To be more precise, we get the following, slightly more
general result.


\begin{restatable}{corollary}{corollaryapproxanwd}
    \label{cor:approx-awnd}
    Let $\cal C$ be an almost nowhere dense class of graphs.
    For every $\varepsilon > 0$, every graph $G \in \mathcal{C}$ and every quantifier-free first-order formula $\phi(y\bx)$, we can compute 
    in time $O(n^{1+\epsilon})$ a vertex tuple $\bu \in V(G)^{|\bx|}$ with \[
        |\max_{\bu} [[\cnt y \phi(y \bu)]]^G - [[\cnt y \phi(y \bu^*)]]^G|\leq n^\epsilon .
    \]
\end{restatable}


Talking about characterizations of almost nowhere dense graph classes,
we provide a plethora of different characterizations, similar to the
ones for bounded expansion and nowhere denseness. We show that a class
is almost nowhere dense classes if and only if measures like $r$-shallow
(topological) minor, forbidden $r$-subdivisions and (weak) $r$-coloring
numbers are bounded by $f(r,\varepsilon)n^\varepsilon$.

We also examine almost nowhere dense classes from an algorithmic point of view:
Whereas it is ``natural'' to consider monotonicity as closure property for
nowhere dense graph classes, it is similarly natural to consider closure under
edge deletion for almost nowhere dense graph classes. Consider a graph class
$\cal C$ which is closed under deleting edges. Then we show that the problem of
finding an $r$ times subdivided $k$-clique is fpt for every fixed~$r$ on $\cal
C$ if and only if $\cal C$ is almost nowhere dense.  In particular, for every
graph class that is not almost nowhere dense, but closed under deletion of
edges, there exists a number~$r$ such that finding $r$-subdivided
$k$-cliques cannot be solved in fpt time under some complexity
theoretic assumption, and, therefore, the FO model checking problem   
for formulas of the form 
$\exists\bx\phi(\bx)$ where $\phi(\bx)$ is quantifier free and has predicates
for adjacency and distance-$r$ adjacency, cannot be solved either.
The situation for distance-$r$
independent set is different:  Like finding an $r$-times subdivided clique it is
fpt on almost nowhere dense graph classes, but there exists a graph class which
is not almost nowhere dense and is closed under edge deletion where the problem
is fpt.

\subsection{Techniques}
For \Cref{thm:algo}, we use a novel dynamic programming technique on game trees
of Splitter games.  Splitter games were introduced by Grohe, Kreutzer, and
Siebertz~\cite{GroheKS17} to solve the first-order model-checking problem on
nowhere dense classes.  Together with their new concept of sparse neighborhood
covers they achieved small recursion trees of constant depth.

Splitter games can be understood as a localized variation of the cops and
robbers game for bounded treedepth (not to be confused with locally bounded
treedepth).  In contrast to \cite{GroheKS17} we apply a dynamic programming
approach, similar to the ones used on bounded tree-depth decompositions.  In
contrast to bounded treedepth, a graph decomposes into neighborhoods of small
radius instead of connected components when removing vertices according to
Splitter's winning strategy. A challenge is that the resulting
neighborhoods---in contrast to connected components---are not disjoint and lead
to double counting for counting problems (an issue that does not occur in
FO-model checking). To avoid double counting we introduce so-called cover
systems specifically for the subgraph ``induced'' by the solution.  The
existence of such cover systems shows that there is a disjoint selection of
small neighborhoods that cover all the vertices relevant to our counting
problem.  By solving a certain variation of the independent set problem, we can
find such a selection and can  safely combine the results of local parts of the
graph as in dynamic programs for bounded tree-depth.

To achieve our second result \Cref{cor:approx-awnd}, we adapt the techniques of the
proof for solving the corresponding exact counting problem on classes of bounded
expansion~\cite{DreierR2021}: We replace $\cnt y \phi(y\bx)$ by a sum of
gradually simpler counting terms until they are simple enough to be easily
evaluated.  During this process we use transitive fraternal augmentations and a
functional representation to encode necessary information into the graph, which
is needed during the above simplification of counting terms.  Along the way some
difficult to handle literals appear in only a few number of terms. Ignoring them
leads to the imprecision of our approximation.  As the number of functional
symbols in (almost) nowhere dense graph classes is not bounded by a constant as
it is the case in classes of bounded expansion, the techniques from
\cite{DreierR2021} have to adapted and extended. The main problem why their
proof cannot be used directly is that the replacement of formulas leads to
formulas of constant size in the case of bounded expansion, but to a
non-constant size in our case. Here we use some new tricks and observe, that
even though the transformed formulas can be of subpolynomial length, they can
basically be replaced by many short formulas.

\section{Preliminaries}

\iflongpaper
\subsection{Graphs.}
We obtain results for \emph{labeled graphs}. A labeled graph is a tuple $G = (V,
E, P_1, \dots, P_m)$, where $V$ is the vertex set, $E$ is the edge set and $P_1,
\dots, P_m \subseteq V$ the labels of $G$. The \emph{order $|G|$} of $G$ equals
$|V|$. We define the \emph{size $||G||$} of $G$ as $|V| + |E| + |P_1| + \dots +
|P_m|$. Unless otherwise noted, our graphs are undirected. For a directed graph
$G$, the indegree of a node $v$ equals the number of vertices $u$ such that
there is an arc $uv$ in $G$. The maximal indegree of all nodes in $G$ is denoted
by $\Delta^-(G)$.

While our results all work for labeled graphs, we will sometimes ignore labels
in long chains of transformations between structures in order to keep the proof
uncluttered.  The presence of labels, however, is never a real problem.

\subsection{Sparse Graph Classes}

A graph $G'$ is an \emph{$r$-subdivision} of a graph $G$ if $G'$ can be obtained
from $G$ by replacing all edges by vertex disjoint paths with \emph{exactly} $r$
inner vertices. Similarly, $G'$ is an \emph{${\le}r$-subdivision} of a graph $G$
if $G'$ is obtained from $G$ by replacing all edges by vertex disjoint paths
with \emph{at most} $r$ inner vertices. Here, the number of subdivisions may
differ for each edge. In $G'$, the vertices of $G$ are called \emph{principal
vertices} and the remaining ones are called \emph{subdivision vertices}. A graph
$G$ is a \emph{topological depth-$r$ minor} of a graph $H$ if an
${\le}r$-subdivision of $G$ is isomorphic to a subgraph of $H$.

\begin{definition}[Bounded expansion]
A graph class $\cal C$ has \emph{bounded expansion} if for all $r \in \N$ there
exists $t \in \N$ such that for all $G \in \cal C$, and all topological
depth-$r$ minors $H$ of $G$, $||H||/|H|\le t$.
\end{definition}

\begin{definition}[Nowhere dense]
A graph class~$\cal C$ is \emph{nowhere dense} if for all $r \in \N$ there
exists a $t \in \N$ such that no $G \in \cal C$ contains $K_t$ as a topological
depth-$r$ minor.  If a graph class is not nowhere dense it is called
\emph{somewhere dense.}
\end{definition}

\fi

\subsection{Weak coloring numbers}
A central concept in this paper are generalized coloring numbers, especially the
weak coloring numbers introduced by Kierstead and
Yang~\cite{coloringdefinition}. An \emph{ordering} $\pi$ of a graph $G$ is a
linear ordering of its vertex set and the set of all such orderings is denoted
by $\Pi(G)$.

\begin{definition}[Kierstead and Yang~\cite{coloringdefinition}]
A vertex $u \in V$ is \emph{weakly $r$-reachable} from a vertex $v \in V$ with
respect to $\pi \in \Pi(G)$ if $u \leq_\pi v$ and there exists a path $P$ from
$u$ to $v$ of length at most $r$ such that $u\leq_\pi w$ for each $w \in V(P)$.
The set of weakly $r$-reachable vertices from $v$ with respect to $\pi$ is
denoted by $\WReach_r[G, \pi, v]$.  Note that $v$ is always included in this
set. We write $\wdist_{G,\pi}(u,v) \leq d$ if $u \in \WReach_r[G, \pi, v]$ or $v
\in \WReach_r[G, \pi, u]$.

The \emph{weak $r$-coloring number} of a graph $G$ (and an ordering $\pi$) is
defined as
\begin{align*}
\wcol_r(G, \pi) &\coloneqq \max_{v \in V(G)} |\WReach_r[G, \pi, v]| \\
\wcol_r(G) &\coloneqq \min_{\pi \in \Pi(G(V))} \wcol_r(G, \pi) .
\end{align*}
\end{definition}

\begin{figure}
\centerline{\includegraphics{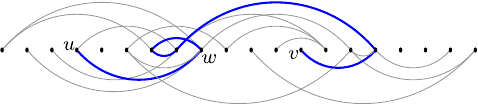}} \caption{$u$ is
weakly $5$-reachable from $v$ by the highlighted path, but $w$ is not weakly
reachable from~$v$.}
\label{fig:wreach}
\end{figure}

The weak $1$-coloring number of a graph is one more than its degeneracy, which
is the smallest number $d$ such that every subgraph $H \subseteq G$ has a vertex
of degree at most $d$ in~$H$. The weak coloring number can be seen as a
localized version of tree-depth, as
\[
\wcol_1(G) \leq \wcol_2(G) \leq \dots \leq \wcol_\infty(G) = \td(G)
~\hbox{\cite{sparsity}}.
\]
Figure~\ref{fig:wreach} contains an example of weak $r$-reachability. Weak
coloring numbers can be used to characterize nowhere dense graph classes:

\begin{proposition}[\cite{Zhu2009,nesetril2011600}]
A graph class $\calC$ is nowhere dense if and only if there exists a function
$f$ such that for every $r \in \N$, every $\varepsilon > 0$, every graph $G \in
\calC$ satisfies $\wcol_r(H) \leq f(r,\varepsilon) |H|^\varepsilon$ for every $H
\subseteq G$.
\end{proposition}

\iflongpaper 


When weak coloring numbers are used within an algorithm, it is often essential
to find an ordering of the vertices of a graph with a small weak coloring
number. The situation is similar to efficient algorithms on tree decompositions:
First a tree decomposition has to be found. Even though computing $\wcol_r(G)$
is NP-hard for $r \geq 3$ in general~\cite{GroheKRSS18}, it is possible to
compute in parameterized quasi-linear time orderings which are approximately
optimal:

    

\def\corgks{\cite[Cor.~5.8]{GroheKS17}}
\begin{proposition}[Grohe, Kreutzer, Siebertz \corgks]
\label{prop:nd-ordering}
Let $\calC$ be a nowhere dense graph class.
There is a function $f$ such that for all $r \in \N$, $\varepsilon > 0$
and $G \in \calC$ with $|G| \geq f(r, \varepsilon)$,
an ordering $\pi$ of $G$ with
$|\WReach_r[G, \pi, v]| \leq |G|^\varepsilon$ for all $v \in V(G)$ 
can be found in time $f(r, \varepsilon) \cdot |G|^{1+\varepsilon}$.
\end{proposition}

As we are dealing with somewhere dense graph classes in \Cref{sec:approx}, we cannot use
Proposition~\ref{prop:nd-ordering} to construct orderings with small generalized coloring numbers.
Revisiting the proof of \Cref{prop:nd-ordering},
we notice that
even without the assumption of nowhere denseness,
one can find in linear time
orderings which approximate the weak coloring numbers:

\begin{lemma}\label{lem:wcol-approx}
There is a computable function $f$ and an algorithm that computes for
a graph $G$ a vertex ordering $\pi$ such that $\wcol_r(G,\pi)\leq
\wcol_{f(r)}(G)^{f(r)}$ for every $r\in\N$.  

The running time of the construction is $\wcol_{f(r)}(G)^{f(r)}|G|$.
\end{lemma}

\begin{proof}
    We use Theorem~4.6.4 of~\cite{Siebertz2016} stating: {\it``For every
    integer $r>0$ there is a polynomial $q_r(x)$ such that for every
    graph~$G$ one can compute in time $q_r\bigl(\nabla_{2^{r+1}}(G)\bigr)$
    an orientation $\vG$ of $G$ and a transitive fraternal augmentation
    $\vG_1\subseteq\ldots\subseteq\vG_r$ with $\vG_1=\vG$ such that
    $\Delta^-(\vG_r)\leq q_r\bigl(\nabla_{2^{r+1}}(G)\bigl)$.''}
    The proof is not contained in~\cite{Siebertz2016}, but follows easily
    from Corollary~5.3 in~\cite{NesetrilM2008a}.  A close look at the
    proof reveals that the running time is indeed
    $q_r\bigl(\nabla_{2^{r+1}}(G)\bigr)n$, where $q_r$ is a polynomial
    that is computable given~$r$.
    
    Using this result yields a digraph $\vG_r$ such that 
    $\Delta^-(\vG_r)\leq q_r\bigl(\nabla_{2^{r+1}}(G)\bigl)$.
    Grohe, Kreutzer, and Siebertz show that if $\vH$ is an $r$-transitive
    fraternal augmention of a graph~$G$ with $\Delta^-(\vH)\leq d$, then
    $\wcol_r(G)\leq 2(d+1)^2$~\cite[Lemma~6.7]{GroheKS17}.  Moreover,
    in the proof of this lemma it is shown that an order can be constructed
    in linear time that witnesses this bound on the weak coloring number.
    
    Hence, we compute a linear order $\pi$ on the vertices of $\vG_r$
    such that
    \begin{equation}\label{equ:wcolnabla}
    \wcol_r(G,\pi)\leq
    2\bigl(q_r\bigl(\nabla_{2^{r+1}}(G)\bigr)+1\bigr)^2.
    \end{equation}
    The grad $\nabla_{(r-1)/2}(G)$ is bounded by the weak coloring number via
    $\nabla_{(r-1)/2}(G)+1\leq\wcol_r(G)$~\cite[Lemma~7.11]{sparsity}.
    Combining this bound with the bound in~(\ref{equ:wcolnabla})
    yields $\wcol_r(G,\pi)\leq2\bigl(q_r(\wcol_{2^{r+2}+1}(G)+1)\bigr)^2
    \leq\wcol_{2^{r+2}+1}(G)^{f(r)}$ for some~$f(r)$.
    
    Altogether we have constructed an ordering $\pi$ in time linear in
    $||\vG_r||\leq\Delta^-(\vG_r)|G|\leq \wcol_{f(r)}(G)^{f(r)}|G|$
    such that $\wcol_r(G,\pi)\leq \wcol_{f(r)}(G)^{f(r)}$.
\end{proof}
\fi

\subsection{Splitter game}
We will use a game-based characterization of nowhere denseness introduced by
Grohe, Kreutzer and Siebertz \cite{GroheKS17}. Given a graph $G$, a radius $r$
and a number of rounds $\ell$, \emph{the $(\ell, r)$-Splitter game} on $G$ is an
alternating game between two players called \emph{Splitter} and
\emph{Connector}. The game starts on $G_0 = G$. In the $i$th round, the
Connector chooses a vertex $v_i$ from $G_i$. Then the Splitter chooses a vertex
$s_i$ from the radius-$r$ neighborhood of $v_i$ in $G_i$. The game continues on
$G_{i+1} = G_i[v_i] - s_i$. Splitter wins if after $\ell$ rounds the graph is
empty. Grohe, Kreutzer and Siebertz showed that nowhere dense graph classes can
be characterized by Splitter games:
\begin{proposition}{\rm \cite{GroheKS17}}\label{prop:game} Let $\C$ be a nowhere
    dense class of graphs. Then, for every $r > 0$, there is $\ell > 0$, such
    that for every $G \in \C$, Splitter has a strategy to win the $(\ell,
    r)$-splitter game on $G$.
\end{proposition}
Note that a winning move of Splitter in a current play can be computed in almost
linear time \cite[Remark 4.7]{GroheKS17}.

\subsection{Sparse neighborhood covers}

Even though the splitter game ends after a bounded number of rounds $\ell$ for nowhere dense classes, the game tree, i.e. the tree spanned by all possible plays of Splitter and Connector, can still be large, e.g. in the dimensions of $n^\ell$. 
To make the game trees small and useful for algorthmic use, Grohe, Kreutzer and Siebertz introduced sparse neighborhood covers \cite{GroheKS17}. These covers group ``similar'' neighborhoods into a small number cluster of bounded radius. These clusters can be used instead of the neighborhoods, reducing the size of the game tree to $O(n^{1+\varepsilon})$.

\begin{definition}\cite{GroheKS17} For a radius $r \in \N$, an
\emph{$r$-neighborhood cover $\mathcal X$} of a graph $G$ is a set of connected
subgraphs of $G$ called \emph{clusters}, such that for every vertex $v \in V(G)$
there is some $X \in \mathcal X$ with $N_r[v] \subseteq V(X)$. The degree of $v$
in $\mathcal X$ is the number of clusters that contain $v$ and the radius of
$\X$ is the maximal radius of a cover in $\X$. A class $\C$ admits \emph{sparse
neighborhood covers} if there exists $c \in \N$ and for all $r \in \N$ and all
$\varepsilon >0$ a number $d = d(r,\varepsilon)$ such that every graph $G \in
\C$ admits an $r$-neighborhood cover of radius at most $c$ and degree at most $d
|G|^\varepsilon$.
\end{definition}

\begin{proposition}\label{prop:nhood-cover} {\rm\cite{GroheKS17}} Every nowhere
dense class $\C$ of graphs admits a sparse neighborhood cover. For a graph $G
\in \C$ and $r \in \N$ such an $r$-neighborhood cover can be computed in time
$f(r,\varepsilon) n^{1+\varepsilon}$ for every $\varepsilon>0$.
\end{proposition}
Indeed, the existence of such covers is another characterization of nowhere
dense classes.

\begin{definition}\label{def:Xr} For a graph $G$ with a vertex order $\pi$, $r
\in \N$ and a vertex $v \in V(G)$, we define $X_r[G,\pi,v]$ as $\{ u \in V(G)
\mid v \in \WReach_{r}[G,\pi,u]\}$. We let $\mathcal{X}_r = \{ X_{2r}[G,\pi,v]
\mid v \in V(G) \}$.
\end{definition}
From the proof of \Cref{prop:nhood-cover} it follows, that the set family
$\mathcal{X}_r$ is such a sparse neighborhood cover.

\subsection{Low treedepth colorings}
A crucial algorithmic tool in the study of bounded expansion and nowhere dense
graph classes are \emph{low treedepth colorings}, also known as $r$-centered
colorings.

\begin{definition}
    An $r$-treedepth coloring of a graph $G$ is a coloring of vertices of $G$ 
    such that any $r' \leq r$ color classes induce a subgraph with treedepth at 
    most $r'$.
\end{definition}

The following statement by Zhu \cite{Zhu2009} is modified such that it is
constructive and holds also for a given vertex ordering $\pi$. It follows from
the original proof.
\begin{proposition}[{\cite[Proof of Thm.\ 2.6]{Zhu2009}}]
    \label{prop:treedepth-col} If $\pi$ is a vertex ordering of a graph $G$ with
    $\wcol_{2^{r-2}}(G, \pi) \leq m$, an $r$-treedepth coloring can be computed
    with at most $m$ colors in time~$O(m n)$.
\end{proposition}

Graph classes of bounded expansion can be characterized by low treedepth
colorings, i.e.,  each graph has an $r$-treedepth coloring with at most $f(r)$
many colors.

\iflongpaper

\subsection{Logic}
We are mainly interested in a small fragment of first-order counting
logic, namely formulas of the form $\cnt y \phi(y \bx)>N$ where $\phi$
is a quantifier-free first-order formula with free variables~$y\bx$
and $N$ is a natural number.

The length of a formula $\phi$ is denoted by $|\phi|$ and equals its
number of symbols, where the length of $N$ counts as one.
All signatures are finite and the cardinality $|\sigma|$ of a signature $\sigma$
equals the number of its symbols.
We often interpret conjunctive clauses $\omega \in \FO$ as a set of literals
and write $l \in \omega$ to indicate that $l$ is a literal of $\omega$.

We denote the universe of a structure $G$ by $V(G)$.
We interpret a labeled graph $G = (V,E, P_1, \dots, P_m)$ as a logical structure with a universe $V$,
binary relation $E$ and unary relations $P_1, \dots, P_m$.

The notation $\bx$ stands for a non-empty tuple $x_1 \dots x_{|\bx|}$.
We write $\phi(\bx)$ to indicate
that a formula $\phi$ has free variables $\bx$. 
Let $G$ be a structure, $\bu \in V(G)^{|\bx|}$ be a tuple
of elements from the universe of $G$, and $\beta$ be the assignment with $\beta(x_i) = u_i$ for $i \in
\{1,\dots , |\bx|\}$. 
For simplicity, we write $G \models \phi(\bu)$ and $[[\phi(\bu)]]^G$
instead of $(G, \beta) \models \phi(\bx)$
and $[[\phi(\bx)]]^{(G, \beta)}$.

The logic \FO is defined in the usual way for functional structures. The
\emph{functional depth} of a formula is the maximum level of nested function
applications, e.g., the formula $f(g(x))=y$ has functional depth $2$. We define
$FO[d, \sigma]$ to be all first-order formulas with functional depth $d$ and
functional signature $\sigma$.

We will both use functional and relational structures, but we will restrict
ourselves to functions of arity one and relations of arity one and two. A
structure $G$ with signature $\sigma$ has \emph{multiplicity} $m$ if for every
distinct pair $u, v \in V(G)$, the number of function symbols $f \in \sigma$
with $u = f_G(v)$ or $v = f_G(u)$ and relation symbols $R \in \sigma$ such that
$R^G(u,v)$ is at most~$m$.

\fi

\section{Exact Evaluation on Nowhere Dense Classes} \label{sec:exact}

In this section we consider the model-checking problem for formulas
$\exists x_1 \dots x_k \cnt y \phi(y \bx) > N$ on nowhere
dense graph classes for quantifier-free first-order formulas $\phi$.
We show that this problem can be solved in almost linear fpt time by
solving its optimization variant $\max_{\bu \in V(G)^{\bx}} \cnt y [[\phi(y \bu)]]$.

\iflongpaper 
\subsection{Replace Formulas with Clauses}
We start with a simplification of the input formula. The quantifier-free formula
$\phi$ is transformed into a set of weighted positive clauses, i.e. formulas
which are conjunctions of positive edge relations with an integer weight
assigned to them. 
The advantage of positive clauses is that each vertex $u$ satisfying $\omega(u \bv)$ is adjacent to a vertex in $\bu$, making the problem very local. 

\begin{lemma}\label{lem:remove_uneq}
    Consider a quantifier-free FO-formula $\phi(y\bx)$ with signature $\sigma$.
    In time $f(|\phi|)$ one can construct a set $\Omega$ with the following
    properties:
    \begin{enumerate}
    \item The set $\Omega$ contains pairs of the form $(\mu, \omega(y \bx))$
    where $\mu \in \Z$ and $\omega(y\bx)$ is a conjunctive clause containing
    only positive literals,
    \item $|\Omega| \leq 4^{|\phi|}$,
    \item $|\omega|\leq|\phi|$ for each $(\mu, \omega) \in \Omega$,
    \item $|\mu| \leq 4^{|\phi|}$ for every $(\mu, \omega) \in \Omega$,
    \item for every graph $G$ and every $\bu \in V(G)^{|\bx|}$, \\
    $\displaystyle[[\cnt y \phi(y \bu)]]^{G} = \sum_{\mathclap{(\mu, \omega) \in
    \Omega}} \mu [[\cnt y \omega(y \bu)]]^{G}.$
    \end{enumerate}
\end{lemma}
%

\begin{proof}
    Let $L$ be the set of literals in $\phi$.
    We construct a formula $\phi'$, equivalent to $\phi$, in disjunctive normal form (disjunction of conjunctions).
    We can assume $\phi'$ to be complete in the sense that every atom in $L$ occurs in every clause of $\phi'$ (either positively or negatively).
    Thus, every clause of $\phi'$ contains $|L| \le |\phi|$ literals.
    For every conjunctive clause $\omega$ of $\phi'$
    we add the tuple $(1,\omega)$ into a set $\Omega$.
    Since by completeness the clauses of $\phi'$ are mutually exclusive,
    \begin{equation}\label{eq:lambdasum}
        [[\cnt y \phi(y\bu)]]^G = \sum_{(\mu,\omega)\in\Omega}
	\mu [[\cnt y \omega(y\bu)]]^G.
    \end{equation}
    Fix a tuple $(\mu,\omega) \in \Omega$. 
    Unless $\omega$ contains only positive literals,
    we can write it as $\omega' \land \neg l$, where $l$ is a positive literal.
    By first ignoring $l$
    and then subtracting what we counted too much we get
    \begin{equation}\label{eq:inclexcl}
    \mu [[\#y\, \omega(y\bx)]]^{\vG}
    = \mu [[\#y\, \omega'(y\bx) \land l]]^{G} 
    - \mu [[\#y\, \omega'(y\bx)]]^{G}.
    \end{equation}
    We remove $(\mu,\omega)$ from $\Omega$
    and add two new entries with conjunctive clauses as
    in~(\ref{eq:inclexcl}) such that $\Omega$ still satisfies~(\ref{eq:lambdasum}).
    Both newly introduced formulas contain one negative literal less.
    It can happen that we want to add some $(\mu,\omega')$ to
    $\Omega$ when $\Omega$ already contains~$(\mu',\omega')$.  In that
    case we replace the latter by $(\mu+\mu',\omega')$.

    We perform this procedure on $\Omega$ until no longer possible.
    The length of each clause in $\Omega$ is still at most $|\phi|$.
    The size of $\Omega$ is at most $4^{|\phi|}$ as the complete DNF formula 
    $\phi'$ has at most $2^{|\phi|}$ clauses of length at most $|\phi|$, and 
    applying the previously described inclusion-exclusion steps exhaustively to one clause results 
    in at most $2^{|\phi|}$ new clauses.

    As the bound for $|\Omega|$ follows from counting the resulting clauses (without deduplicating possible duplicates), the same bound of $4^{\phi}$ also follows for the weights.
\end{proof}   
\fi

\subsection{Radius-$r$ Decomposition Tree}
In the following, we will introduce a new kind of decomposition, which heavily
relies on the ideas from \cite{GroheKS17}. We call it the radius-$r$
decomposition tree. For illustration, consider a tree-depth decomposition of a
graph $G$. It has the property that after the removal of the root $v$ in the
decomposition, for each connected component $C$ of $G-v$ there exists a child of
$v$ in the decomposition that contains $C$. In the radius-$r$ decomposition
tree, not every connected component is represented by a child but every
radius-$r$ neighborhood of $G-v$ instead. Another difference is that these
neighborhoods are not necessarily disjoint. We will use this radius-$r$
decomposition tree as the structure on which a dynamic program will solve
$\max_{\bu} [[\cnt y \phi(y \bu)]]^G$.


\begin{definition}\label{def:our-tree}
Let $G$ be a graph. Let $r,\ell \in \N$ be such that splitter has a winning
strategy for the $\ell$-round radius-$2r$ splitter game on $G$. Let $\pi$ be an
ordering of $G$.

A \emph{radius-$r$ decomposition tree $T_r(G,\pi,\ell)$} is a pair $(T, \beta)$
where $T$ is a tree of depth $\ell$ and $\beta \colon V(T) \to V(G)$. We
construct it recursively. If $G$ is empty, $T_r(G,\pi,\ell)$ is the empty tree. 

Let $s \in V(G)$ be the first move of the winning strategy of splitter for the
$(\ell, 2r)$-splitter game on $G$. The root is a node $t$ with $\beta(t) = s$.
For every $v \in V(G)$ we append the decomposition tree $T_r(G[X_v], \pi,
\ell-1)$ where $X_v = X_{2r}[G-s, \pi, v]$.
\end{definition}
Note that the case $\ell = 0$ while the graph is not empty, cannot happen due to
the Splitter having a winning strategy. 

\begin{corollary}
Let $G$ be a graph, $\pi$ a vertex ordering of $G$, $r, \ell \in \N$ and $T =
T_r(G,\pi,\ell)$ a radius-$r$ decomposition tree. Let $t \in V(T)$ be a node and
$T_t$ be the subtree of $T$ starting at $t$. Then for every $u \in W :=
\beta(V(T_t)) \setminus \{\beta(t)\}$ there exists a child $t'$ of $t$ such that
$N_r^{G[W]}[u] \subseteq \beta(T_{t'})$.  
\end{corollary}
\iflongpaper
As $\cal X_r = \{ X_{2r}[G,\pi,v] \mid v \in V(G) \}$ is by
\Cref{prop:nhood-cover} a radius-$r$ cover, the fact follows immediately.
\fi

\begin{lemma}\label{lem:tree-good} Let $G$ be a graph, $\pi$ a vertex ordering
of $G$ and $r, \ell \in \N$. Then, the radius-$r$ decomposition tree $T =
T_r(G,\pi,\ell)$ (\Cref{def:our-tree}) has size $|T|\leq \wcol_{2r}(G,\pi)^\ell
n$ and depth~$\ell$. The construction time is linear in $|T|$.
\end{lemma}

\begin{proof}
    By construction, the depth of the tree is determined by the depth of the
    splitter game, which is $\ell$.
    
    Consider the root path $P_t$ of some node $t \in V(T)$. Then $\beta(P_t)
    \subseteq \WReach_{2r}[G,\pi, \beta(t)]$. As the length of $P_t$ is at most
    $\ell$, $\beta(t)$ appears at most $\WReach_{2r}[G,\pi, \beta(t)]^\ell \leq
    \wcol_{2r}(G,\pi)^\ell$ times (as a $\beta$-label of nodes) in $T$. Thus,
    $|T| \leq  \wcol_{2r}(G,\pi)^\ell n$.
\end{proof}

\begin{corollary}
Let $\cal C$ be a nowhere dense graph class. For every $r \in \N$ the
$r$-decomposition tree has constant depth, almost linear size and can be
computed in almost linear time. 
\end{corollary}

\subsection{Cover Systems}

Given a subgraph $H$ in $G$ with a vertex ordering $\pi$ of $G$. A \emph{cover
system} of $H$ in $G$ is a family $\cal Z$ of clusters $Z_i = X_r[G,\pi,v] \in
Z$ for some $r \in \N$ such that every connected component $C$ of $H$ is
contained in some~$Z_i$. A cover system is \emph{non-overlapping} if all
distinct clusters have an empty intersection.
\begin{lemma}\label{lem:ball2} For every graph $G$ with a vertex ordering $\pi$,
    every $D \subseteq V(G)$ of size $k$, there exists a cover system of
    $G[N[D]]$ in $G$ of size at most $k$ where each cluster has the same radius
    $r \leq 2^k$.
\end{lemma}

\begin{proof}
    We start with the clusters $X_2[G,\pi, \min_\pi N[d]]$ for every $d \in D$.
    Call this collection $\cal Z$. Note that $\cal Z$ is already a valid cover
    system of $G[N[D]]$ in $G$. If two distinct clusters $X_{r}[G,\pi, z]$ and
    $X_{r}[G,\pi, z']$ from $\cal Z$ intersect, we replace both with a new
    cluster $X_{2r}[G,\pi, \min_\pi\{z,z'\}]$ in $\cal Z$. Every vertex or edge
    covered by the two old clusters stays covered in the new one. Also, if two
    clusters $X_{r}[G,\pi, z]$ and $X_{r'}[G,\pi, z']$ are of a different
    radius, say, $r' < r$, we replace $X_{r'}[G,\pi, z']$ with $X_{r}[G,\pi,
    z']$ to match the radii of all the clusters.

    We repeat this until no intersecting clusters remain. As the number of
    clusters decreases with every step, the radius is at most $2^k$ at the end.
\end{proof}

For \Cref{thm:algo}, one needs to find clusters from $\cal X_r$ which 
are disjoint and maximize the sum of weights of clusters. This is captured by
the following definition. We can solve this problem in almost linear time on
nowhere dense graph classes, by noticing that the intersection graphs of the
sparse neighborhood covers $\cal X_r$ are almost nowhere dense. Then, one can
use treedepth colorings and LinEMSOL.

\begin{definition}[Disjoint Cluster Maximization] \label{prob:disj-clusters}
Given a graph, a set system $\cal X_r$ as defined in \Cref{def:Xr}, labelled by
a function $\Lambda : \cal X_r \to 2^\Lambda$ of size~$k$. Each combination of a
cluster $X \in \cal X_r$ and label $\lambda \in \Lambda(X)$ is weighted by a
function $w$. 

Problem: Find pairwise disjoint clusters $X_1, \dots, X_k \in \cal X_r$ such
that for each label $\lambda_i \in \Lambda$ the cluster $X_i$ is labeled
$\lambda_i$ and $X_1, \dots, X_k$ maximize $\sum_{i=1}^k w(X_i,\lambda_i)$ for
such cluster sets.

Parameter: $r, k$
\end{definition}

\iflongpaper
\begin{lemma} \label{lem:sparse_intersections}
  Let $\cal C$ be a nowhere dense class of graphs and $r \in \N$. Then there
  exists an almost nowhere dense graph class $\cal I$ such that for every graph
  $G \in \cal C$, the intersection graph $I$ of $\cal X_r$ (defined in
  \Cref{def:Xr}) is contained in $\cal I$.
\end{lemma}

\begin{proof}
    Assume $\prec$ witnesses a good order in $G$.
    We build a new order $\prec_I$ for $I$.
    $X_r[v]$ is a shorthand for $X_r[G,\prec, v]$.
    We say $X_r[v] \prec_I X_r[u]$ if $v \prec u$. 
    Then 
    \begin{align*}
    \WReach_s[I, \prec_I, X_r[u]]  &= \{ Y \in \mathcal X \mid
        \text{path $P = Y_0 \dots Y_s$ of length at most $s$ in $I$}, \\
        &\qquad Y_0 = X_r[u], Y_s = Y = \min_{\prec_I} P \} \\
    &\subseteq \{X_r[v] \in \mathcal X \mid  v \in \WReach_{2rs}[G, \prec, u]\}
    \end{align*}
    Note that $v = \min_\prec X_r[v]$. Hence, the last equation follows. 
    As $\WReach_{2rs}[G, \prec, u] \leq n^\epsilon$, so is 
    $\WReach_s[I, \prec_I, X_r[u]]$. Thus, $\cal I$ is  almost nowhere dense. 
\end{proof}

\begin{remark}
  Note that this result cannot be improved to a nowhere dense class of intersection graphs for $\cal X_r$. 
  However, maybe there exists another sparse neighborhood cover whose intersection graph is nowhere dense.
  
  Example: Consider the class of graphs with an independent set of size $n$ with a star of size of $\log n$. For the weak color ordering, order the apex to the right (this is not optimal but the weak coloring number of this ordering is $\log n$).
  The resulting intersection graph contains then a clique of size $n$.
  Hence, the which is somewhere dense. 
\end{remark}

\begin{lemma}
  \label{lem:find-disj-clusters}
  We can solve the Disjoint Cluster Maximization problem in almost linear FPT
  time on nowhere dense class of graphs. 
\end{lemma}

\begin{proof}
    As $G$ is from a nowhere dense graph class, we can apply
    \Cref{lem:sparse_intersections}, yielding a graph $H$ from an almost nowhere
    dense graph class. The labels and weights from $G$ are also added to $H$.
  
    Two clusters $X,Y \in \cal X_r$ are disjoint in $G$ if and only if $X$ and
    $Y$ are not adjacent in $H$. Hence, the original problem on $G$ is
    equivalent to finding an independent set $S$ of size in $H$, where $\sum_i
    w(S_i,\lambda_i)$ is maximized.
  
    Since $H$ is from an almost nowhere dense graph class, by
    \Cref{prop:treedepth-col} there exists a $k$-treedepth coloring of $H$ using
    $n^\epsilon$ many colors. As the optimal independent set $S$ with the
    constraints from above has size $k$, it has to be contained in the subgraph
    of $H$ induced by some selection of at most $k$ colors. Thus, for each
    selection of $k$ colors, we consider the graph~$H'$ induced by those which
    has treedepth at most $k$. As this independent set variation can be
    expressed as MSO-formula and the objective function is linear, we can use
    LinEMSOL on $H'$ to solve this problem optimally. The solution for $H$ is
    then the maximum over the solutions of all $H'$s.
  
    Applying \Cref{lem:sparse_intersections} takes almost linear time. There are
    ${n^\epsilon\choose k} \leq n^{\epsilon'}$ many color combinations and each
    iteration of LinEMSOL takes linear FPT time. 
\end{proof} 
\fi

Let $\Omega$ be the set of weighted positive conjunctive clauses $(\mu,
\omega(y\bx))$, $\bz \subseteq \bx$ and $\bu \in V(G)^{|\bx|}$. With
$\Omega|_{\bz}$ we denote a subset of $\Omega$ with weighted clauses $(\mu,
\omega(y\bx))$ where every variable occurring in $\omega$ is from $\bz$. We
define $\Omega|_{\bz}[Z,\bu]$ as $\sum_{v \in Z} \sum_{(\mu,\omega) \in
\Omega|_{\bz}} \mu [[\omega(v \bu)]]^G$. Note that $\Omega|_{\bz}[Z,\bu]$
depends only on the assignment of $\bz$ and does not need the full assignment
$\bu$ of $\bx$.

To illustrate the following lemma, consider a positive conjunctive clause
$\omega(y\bx\bz)$, sets $P,W \subseteq V(G)$ and $\bu \in P^\bx, \bw \in W^\bz$.
To count the fulfilling vertices $v \in W$ of $\omega$, i.e. $\Omega[W,\bu]$, we
want to reduce this task to counting on cover systems of $N[\bw]$. However, as
not all fulfilling vertices in $W$ are adjacent to $\bw$, we need to be more
careful.

\begin{lemma} \label{lem:fulfilling-distr} Let $G$ be a graph, $\Omega$ a set of
    weighted positive conjunctive clauses $(\mu, \omega(y\bx\bz))$, $P,W
    \subseteq V(G)$ disjoint, $\bu \in P^{\bx}, \bw \in W^{\bz}$ such that
    $N[\bw] \subseteq P \cup W$. For every cover system $\cal Z$ of
    $G[N[\bw]]$ in $G[W]$ it holds that 
    \[ \Omega[W,\bu\bw] = \Omega|_{y\bx}[W,\bu\bw] + \sum_{Z
    \in \cal Z} (\Omega|_{y\bx\bz_Z}[Z, \bu\bw] - \Omega|_{y\bx}[Z, \bu]) \] 
    where $\bz_Z$ are
    the variables $z_i$ from $\bz$ which are assigned to a vertex in $Z$. 
\end{lemma}
\iflongpaper
\begin{proof} 

    For $u \in W$, $G \models \omega(u\bu\bw)$ only if $u$ is adjacent to some
    vertex from $\bu\bw$, as $\omega$ is a positive conjunctive clause. Hence,
    $\Omega[W \setminus N[\bw], \bu\bw] = \Omega|_{y\bx}[W \setminus N[\bw], \bu\bw]$.
    This also holds for $N \setminus \bigcup \cal Z$ instead. By the same
    observation $\Omega[Z, \bu\bw] = \Omega|_{y\bx \bz_Z}$ for $Z \in \cal Z$. We
    get the equality as a result of the observation above and subtracting
    $\Omega|_{y\bx}[Z, \bu\bw]$ to prevent counting vertices in $Z$ twice. 
\end{proof}
\fi
Let us consider how a solution $\bu$ for $\cnt y \phi(y\bx)$ interacts with a
radius-$r$ decomposition of the input graph $G$ where $r$ is chosen
appropriately big, e.g. $2^k$ resulting from \Cref{lem:ball2}. First, we
transform $\phi$ into a set of positive clauses $\Omega$%
\iflongpaper
 using \Cref{lem:remove_uneq}%
\fi
, making the application of \Cref{lem:fulfilling-distr}
possible.

Consider some node $t$ in $T_r$. When applying \Cref{lem:fulfilling-distr} with
$P$ as the vertices of the root path of $t$ and $W$ as $T_t$, we see that the
resulting cover system $\cal Z$ corresponds to a selection of children of $t$ in
$T_r$, as both use the sets $X_r$ from \Cref{def:Xr}. Now imagine that we know
$\Omega_{y\bx\bz_Z}[Z,\bu]$ for every $Z \in \cal Z$. Note that this number only
depends on the assignment of $\bx\bz_z$ and not the vertices assigned outside
$P$ and $Z$. With \Cref{lem:fulfilling-distr} we can combine these numbers into
$\Omega[W,\bu]$ without needing to know the actual assignments of $\bz_Z$ in the
cover system anymore! Note that $\Omega_{y\bx}[Z]$ is easily computable while
only knowing $\bu$ and not $\bw$. 

Thus, we can compute $[[\cnt y \phi(y \bu)]]$ bottom-up using the radius-$r$
decomposition while only considering the vertices assigned in $\bu$ which are
contained in the root path of the considered vertex. 

\subsection{Dynamic Program}
To determine $\max_{\bu} \cnt y \phi(y \bu)$ for a quantifier-free formula
$\phi(y\bx)$ we recursively compute the following information in the
decomposition tree of $G$ (bottom-up, if you will). Consider some node $t$ of
$T$ and a partial assignment $\alpha$ of $\bx$ to the root path $\beta(P_t)$.
The interesting information is: How many vertices underneath $t$, i.e. in
$V(G_t)$, fulfill $\phi$ under the ``best'' choice on completing the assignment
$\alpha$ to vertices in $V(G_t)$. Then the answer to the problem can be read off
the information for the root node.

Assume we already know this kind of information for every child $t'$ of $t$. To
compute this information for $t$, we branch how the variables $x_i$ that are
not assigned under $\alpha$ are distributed among the children of $t$. Then the
table entries of these children are combined in a suitable way. We do this for
every distribution among children and take the maximum of the resulting values.
If a vertex corresponding to $t$ fulfills with the assignment the formula
$\phi$, it gets counted towards the number of ``fulfilling'' vertices.

However, we have to take more into consideration. First, branching on the
distribution of the unassigned variables $x_i$s under $\alpha$ among the
children of $t$ is not fast enough, as there are around $n^k$ possibilities for
that. Instead, we branch on how the unassigned variables are partitioned. For
every such partition, we formalize the optimal choice of children $t_i$ such
that they contain exactly the unassigned variables from the $i$-th part, as an
optimization problem. 

Secondly, the graphs $G_{t'}$ spanned by each child $t'$ of $t$ are in general
not disjoint. Combining the counts of two overlapping graphs yields to double
counting. We circumvent this in the above optimization problem.

Thirdly, we need to keep track of how the vertices in the root path $P_t$ are
adjacent to the variables $x_i$ that are assigned underneath $t$. We cannot
branch on the complete assignment as the number of those is too high.

Before we turn to the dynamic program on the decomposition tree, we consider
something simpler:

Let $G$ be a graph and $\phi(y\bx)$ be quantifier-free FO formula.
Consider the pair $(P,W)$ which is a set of vertices $P = \{v_1, \dots, v_k\}
\subseteq V(G)$ and a set $W \subseteq V(G)$ that is disjoint with $P$. We are
interested in how many vertices $v$ in $G[P \cup W]$ satisfy $\phi(v \bu)$ for
an optimal choice of $\bu \in (P\cup W)^{|\bu|}$. For this, we keep track of
$M_\alpha^{(P,W)}[S]$, which is the number of fulfilling vertices $v \in W$ wrt.
$\phi$, $\hat\alpha$ and $S$, maximizing over $S$-completions $\hat \alpha$ on
$W$.

We can ``forget'' a vertex $v$, i.e., derive the information of $(P,W \cup
\{v\})$ from the information $(P \cup \{v\}, W)$ as follows: Assume the maximum
number of fulfilling vertices in $W$ is $x$ for a given partial assignment
$\alpha$ on $P\cup\{v\}$ and adjacency profile $S$ on $P \cup \{v\}$. Then the
number of fulfilling vertices in $W \cup \{v\}$ is $x+1$ if $v$ satisfies $\phi$
with the assignment $\alpha$ and adjacency profile $S$, or $x$ otherwise.
However, neither $\alpha$ nor $S$ are valid assignments or adjacency profiles
for $P$. Hence, we need to adjust these so that we can formulate this
information for $(P,W\cup\{v\})$. For this, we need to remove $v$ from $\alpha$
and add the neighborhood of $v$ in $P$ to $S$ as $S_i$, for every $i$ with
$\alpha(x_i)=v$. Then, $M_\alpha^{(P \cup \{v\},W)}[S] = M_{\alpha|_P}^{(P, W
\cup \{v\})}[S'] (+1)$ where $\alpha|_P$ is the assignment $\alpha$ without $v$
and $S'$ is the adjacency profile as described above. 

One can also combine the information of two structures $(P, W_1)$ and $(P,W_2)$
to get the information of $(P, W_1 \uplus W_2)$ if $W_1$ and $W_2$ are disjoint.
This is also known as ``merge.'' Consider some assignment $\alpha$ on $P$ and
some adjacency profile $S$ on $P$. Then the number of fulfilling vertices in $U
\uplus W$ wrt $\phi$, $\alpha$ and $S$ is the $\max \{M_\alpha^{P,W_1}[S_1] +
M_\alpha^{P,W_2}[S_2] \mid S_1 \uplus S_2 = S \}$.

Indeed however, the algorithm does not take a quantifier-free formula $\phi$ but
a set of weighted positive conjunctive clauses. Instead of just counting the
fulfilled vertices, it computes the added up weight of them wrt. to the weights
of the clauses. 

\iflongpaper

\subsubsection{Some definitions}

Let $I \subseteq \N$. An \emph{$I$-adjacency profile $S$} of a set $P$ is a
collection of sets $\{S_i \subseteq P \}_{i \in I}$. For $J \subseteq I$, we
denote with $S|_J$ the collection $\{S_i\}_{i \in J}$. Equivalently, $S$ can be
interpreted as a function $S \colon P \to 2^I$.

Let $G$ be a graph and $T$ its $r$-decomposition, $t$ a node in $T$, $\alpha$ a
partial assignment of $\bx$ on $P_t$ and an $I$-adjacency profile $S$ on
$P_t$, where $I \subseteq [|\bx|] \setminus \dom(\alpha)$. Then $\hat \alpha$ is
an \emph{$S$-refinement of $\alpha$} if there exists a partial assignment $\bar
\alpha$ of $(x_i)_{i \in I}$ on $V(G_t)$ and $N(\bar \alpha(x_i)) \cap P_t =
S_i$ for all $i \in I$ and $\hat \alpha = \alpha \uplus \bar \alpha$.

Let $\psi(\bx)$ be a conjunctive clause and $\bz \subseteq \bx$. The
\emph{$\bz$-projection of $\psi$}, $\psi \star \bz$, is the conjunctive clause
which contains the literals of $\psi$ that involve at least one variable of
$\bz$. Note that for a graph $G$ and $\bu \in V(G)^{\bx}$ $G\models \psi(\bu)$
if and only if $G \models (\psi \star x_i) \psi(\bu)$ for all $x_i \in \bx$.

We say $v$ \emph{fulfills} a positive conjunctive clause $\omega(y\bx)$ wrt. a
partial assignment $\alpha$, an $I$-adjacency profile $S$ and a complete
conjunctive clause $\xi(\bx)$ if for every literal $E(y x_i)$ in $\omega$ either
$v$ is adjacent to $\alpha(x_i)$ (if assigned) or $v \in S_i$ (if $S_i$ exists).
The \emph{weight of $v$ in $\Omega$} wrt. $\alpha$ and $S$ is
\[ \sum \{ \mu \mid (\mu, \omega) \in \Omega \text{ and $v$ fulfills $\omega$
wrt. $\alpha$, $S$}\} . \]

Let $P,W \subseteq V(G)$, $v \in P$ such that $N^G[W\cap\hat\alpha(\bx)]
\subseteq W \cup P$. Consider an assignment $\hat\alpha(\bx)$, complete
conjunctive clause $\xi(\bx)$ with $G \models \xi(\hat\alpha(\bx))$. Then
$\Omega[v,\hat\alpha(\bx)] = \sum \mu[[\omega(v\hat\alpha(\bx))]] = \Omega[v,
G,\alpha, S,\xi]$.


With all the tools at hand, we can formulate Algorithm \ref{alg:dynprog} and
show its correctness.  

\begin{lemma} \label{lem:algo-correct} Let $\xi(\bx)$ be a complete conjunctive
    clause, $\Omega_\xi$ be a set of weighted complete conjunctive clauses
    $(\mu, \omega(y\bx))$ where $y$ appears in every literal. Let $G$ be a graph
    and $T$ be a radius-$r$ decomposition of $G$ with $r=2^k$. Then Algorithm
    \ref{alg:dynprog} computes 
    \[ \max_\bu \sum_{(\mu,\omega') \in \Omega_\xi}
    [[\cnt y \omega' (y\bu) \land \xi(\bu)]]^G . \] 
\end{lemma}

\begin{proof}
    Notice that $M$ maps adjacency profiles of $P_t$ to integers. 
    Let $S$ be an $I$-adjacency profile on $P_t$ for some $I$.
    At the end of the recursive call of $algo(t,\alpha)$,
    For every $S$, there exists an $S$-refinement $\hat \alpha$ to $G_t$ 
    such that $M[S] = \Omega[G_t, \hat \alpha( \bx)]$ and $G[G_t \cup P] \models
    (\xi \star I) (\hat\alpha (\bx))$. Let $\hat \alpha$ be such a refinement
    that maximizes $M[S]$.
    
    
    Let $v = \beta_T(t)$.    
    Assume $t$ is a leaf. Then, $v$ can take the multiple roles of any
    unassigned $x_i$ under $\alpha'$ or no role at all. Assume $\alpha'(x_i) =
    v$. Then $S_i$ is the neighborhood of $v$ on $P_t$ and $S = \{S_i \mid
    \alpha'(x_i) = v\}$. In any case, $v$ fulfills $\omega$ wrt. $G$, $S$ and
    $\alpha$ if and only if $v$ fulfills $\omega$, $G$, $\emptyset$ and
    $\alpha'$ for all $\omega \in \Omega$. This is computed in lines
    \ref{line:start-leaf}-\ref{line:end-leaf}.
    
    Otherwise, assume $t$ is an internal node of $T$.
    
    Consider an $I$-adjacency profile $S$ of $P_t$ and et $\hat \alpha$ be an
    $S$-refinement of $\alpha$ to $G_t$ which maximizes the weight of vertices
    in $G_{t}$ wrt. $G$ and $\hat\alpha$, i.e $\Omega[G_t, \hat\alpha(\bx)]$.
    
    We change our viewpoint from $P_t$ to $P_t \cup \{v\}$. For this, let
    $\alpha'$ be the restriction of $\hat \alpha$ to $P_t \cup \{v\}$. In
    another words, $\alpha'$ is an refinement of $\alpha$ to $P_t \cup \{v\}$.
    Let $S'$ be the adjacency profile on $P_t \cup \{v\}$ of $\hat\alpha$, i.e.,
    $S_i' = N[\hat\alpha(x_i)] \cap (P_t \cup \{v\})$ for $\hat \alpha(x_i) \in
    V(G_t - v)$.
    
    To determine $\Omega[G_t-v, \hat \alpha(\bx)]$ we want to apply
    \Cref{lem:fulfilling-distr}: Setting $\bx' = \dom(\alpha') \subseteq
    \beta(P_t \cup \{v\})$, 
    by \Cref{lem:fulfilling-distr} there exists a cover system $\cal Z$ in
    $G_t-v$ of radius $r$ such that  
    \begin{align}\label{eq:omega}
    	\Omega[G_t-v, \hat \alpha(\bx)] = \Omega|_{y\bx'}[G_t-v, \hat\alpha(\bx)] 
    	+ \sum_{Z \in \cal Z} (\Omega|_{y\bx'\bz_Z}[Z, \hat\alpha(\bx)] - \Omega|_{y\bx}[G_t-v, \hat\alpha(\bx)]) .
    \end{align}
    
    Note that by definition $M_{\alpha'}[S] = \Omega[G_t-v, \hat\alpha(\bx)]$
    which equals $\max_\bw \Omega[G_t-v, \alpha'(\bx)\bw]$ where $\bw$ ranges
    over tuples $\bw$ whose sets neighborhoods equals $S$. Both
    $\Omega|_{y\bx'}[G_t-v,\hat\alpha(\bx)]$ and
    $\Omega|_{y\bx'}[Z,\hat\alpha(\bx)]$ can be easily computed in linear time
    (without recursion) as their evaluation depends only on $\alpha'$.
    
    To compute the above sum (\Cref{eq:omega}), we need to determine
    $\Omega|_{y\bx'\bz_Z}[Z, \hat\alpha(\bx)]$ recursively. Consider $Z \in \cal
    Z$ and let $S' = \{S_i \in S \mid \bw_i \in Z \}$. As the covering system
    $\cal Z$ from \Cref{lem:fulfilling-distr} is a subset of $\cal X_r$, by
    construction (\Cref{def:our-tree}) there exists a child $t'$ of $t$ with
    $V(G_{t'}) = Z$ and for that by induction, $M_{t'}[S'] =
    \Omega|{y\bx\bz_{Z}}[Z ,\hat\alpha(\bx)]$.
    
    Finding such a cover system $\cal Z$ for an optimal choice of
    $\hat\alpha(\bx)$ is modeled with an instance of the Disjoint Cluster
    Maximization problem where the weights are set as described in
    \Cref{eq:omega} (lines \ref{line:combine_start}-\ref{line:combine_end}). As
    the form of the cover system is not known beforehand, i.e., it is not know
    which $x_i$ belong into the same cover system, the algorithm branches over
    all partitions of unassigned variables.
    
    To recap, before line \ref{line:forget_start} $M_{\alpha'}[S'] =
    \Omega[G_t-v, \hat\alpha(\bx)]$ where $S'$ is an adjacency profile on $P_t
    \cup \{v\}$. Now note that at this point it is not guaranteed that $\hat
    \alpha(\bx)$ does not contradict $\xi$, i.e., $G \models (\xi \star I)(\hat
    \alpha(\bx))$. By induction, we know that $G \models (\xi \star \bz_Z)(\hat
    \alpha(\bx))$ for all $Z \in \cal Z$. Hence, for the algorithm it remains to
    make sure whether $G \models (\xi \star J)(\hat \alpha(\bx))$ for the
    variables $J = \hat\alpha^{-1}(v) = \alpha'^{-1}(v)$. This can be derived
    from $S'$ and $\alpha'$ and happens in line \ref{line:contradicts}.
    
    After line \ref{line:forget_end}, $M_{\alpha'}[S] = \Omega[G_t,
    \hat\alpha(\bx)]$ where $S$ is an adjacency profile on $\beta(P_t)$ (instead
    of $\beta(P\cup\{x\})$ as before). 

    As now all information about $v$ is taken care of, the parts of the
    assignment which are assigned to $v$ are forgotten and collect the resulting
    values into $M[S]$ (lines \ref{line:collect_start}-\ref{line:collect_end}).

    If $t$ is the root of $T$, we return $M[\varnothing]$ (which is the only
    entry of $M$) which is $\max_\bu \Omega[\emptyset, \bu] = \max_\bu
    \sum_{(\mu,\omega) \in \Omega} \mu [[\cnt y \omega(z\bu)]]^G \land
    [[\psi(\bu)]]^G$.
\end{proof}

\begin{algorithm}
    \footnotesize
    \caption{${\it algo}(t,\alpha)$}\label{alg:dynprog}
    \KwIn{A graph $G$ with a decomposition $T$ of $G$, a node
      $t$ of $T$, a partial assignment $\alpha$ of $\bx$ on $P_t$,  
      complete conjunctive clause $\xi(\bx)$
      and a set $\Omega$ of weighted positive clauses $\omega(\bx)$ 
      }  
      \KwOut{$M$ with $M[S]$ as described above.} \BlankLine
  
    $M,M_{t'} :=$ are empty associative arrays over the family of subsets of
    $\beta(P_t)$ for every child $t'$ of $t$. If an entry is not in the array
    its value is $-\infty$\; $v := \beta_T(t)$ (vertex of $t$)\;
    \tcc{Base case}
    \If{$t$ is a leaf in $T$\label{line:start-leaf}}{
        
      \ForEach{Possible refinement $\alpha'$ of $\alpha$ to $v$}{
        \lIf{$\alpha'$ and $S$ contradict $\xi$ }{skip}
        $S := \{\}$\;
        \ForEach{$i \in \alpha'^{-1}(v)$}{
          $S_{i} := N[v] \cap \beta(P_t)$\; 
          $S := S \cup \{ S_i \}$\;
        }
        $M_{\alpha}[S] := \Omega_{y \dom(\alpha')}[v, \alpha'(\bx)]$
      }
      \Return $M$\;\label{line:end-leaf}
    }
  
    \ForEach{Possible refinement $\alpha'$ of $\alpha$ to $v$\label{line:try-assignments}}{
      clear $M_{t'}$s\;
      \ForEach{Child $t'$ of $t$\label{line:start-x-notin-d}}{
        $M_{t'} := algo(t', \alpha')$\;
      }
      \tcc{combine results from children}
      \ForEach{$I \subseteq [k]\setminus \dom(\alpha')$\label{line:combine_start}}{ \tcc{Not assigned $x_i$s}
        \ForEach{$I$-adjacency profile $S$ on $P_t \cup \{v\}$}{
          \ForEach{Partition $\cal I$ of $I$}{
            Init $w \colon \cal X_r^{G_t} \times \cal I \to \N$ \tcc{$w$ is a weighting function}
            \ForEach{Child $t'$ of $t$ and $J \in \cal I$}{
              $\delta := \Omega|_{y \dom(\alpha')}[G_{t'}, \alpha'(\bx)]$\;
              $w(V(G_t), J) := M_{t'}[S|_J]$\;
            }
            $\Delta := \Omega|_{y \dom(\alpha') J}[G_t-v, \alpha'(\bx)]$\;
            $d^* := \Delta + $ weight of Disjoint Cluster Maximizer of $\cal X_r^{G_t}$ and $w$\; 
            $M_{\alpha'}[S] := \max\{M_{\alpha'}[S], d^*\}$\label{line:combine_end}\;
          }  
        }
      }
      \tcc{forget $v$}
      \ForEach{$S \in M_{\alpha'}$\label{line:forget_start}}{
        \lIf{$\alpha'$ and $S$ contradict $\xi$\label{line:contradicts}}{remove $S$ from $M_{\alpha'}$ and skip}
        $M_{\alpha}[S] :=$ weight of $v$ in $\Omega$ wrt. $\alpha'$ and $S$\;
        $S' := S$\;
        \tcc{add the n'hood of $v$ to adjacency profile with index of $x$ in part. assignment $\alpha'$} 
        \ForEach{$i \in \alpha'^{-1}(v)$}{
          $S_{i} := N[v] \cap \beta(P_t)$\; 
          $S' := S' \cup \{ S_i \}$\;
        }
        $M_{\alpha'}[S'] := M_{\alpha'}[S]$\;
        \lIf{$S \neq S'$}{remove $S$ from $M_{\alpha'}$\label{line:forget_end}}
      }
    }
    
    \tcc{collect}
    \ForEach{adjacency profile $S$ (without $v$)\label{line:collect_start}}{
      $M[S] = \max \{ M_{\alpha'}[S] \mid \text{$\alpha'$ is an refinement of $\alpha$ on $v$}  \}$\label{line:collect_end}\;
    }
  
    \tcc{return}
    \lIf{$v$ is the root of $T$}{\Return $M[\varnothing]$}
    \lElse{\Return $M$}
 \end{algorithm}

\fi

\theoremalgo*

\iflongpaper

\begin{proof}
    Using \Cref{lem:remove_uneq} we can compute a set of weighted positive
    conjunctive clauses $\Omega$ with 
    \[
        [[\cnt y \phi(y\bu^*)]]^G = \sum_{(\mu,\omega) \in \Omega} [[\cnt y \omega (y\bu^*)]]^G
    \] for every $\bu \in V(G)^\bx$ in time $f(k)$.

    For every complete conjunctive clause $\xi(\bx)$, we compute the set
    $\Omega_\xi$. Let $\omega(\bx)$ be a conjunctive clause. We decompose
    $\omega$ into $\omega(\bx) \equiv \omega'(y\bx) \land \psi(\bx)$ where
    $\psi(\bx)$ is the conjunction of literals of $\omega$ which contain only
    $\bx$ as variables and $\omega'(y\bx)$ are remaining literals of $\omega$.
    For every $(\mu, \omega) \in \Omega$, $(\mu, \omega'(y\bx))$ is added to
    $\Omega_\xi$ where $\omega(y\bx) \equiv \Delta(y\bx) \land \psi(\bx)$ as
    above and $\xi(\bx) \models \psi(\bx)$. Note that for every vertex tuple
    $\bu$ there exists exactly one such $\xi$ with $G \models \xi(\bu)$. Also,
    for that $\xi$ \[ \sum_{(\mu,\omega) \in \Omega} [[\cnt y \omega (y\bu)]]^G
    = \sum_{(\mu,\omega') \in \Omega_\xi} [[\cnt y \omega' (y\bu)]]\land
    [[\xi(\bu)]]^G . \]

    Computing a good ordering $\pi$ of $G$ with $\wcol_r(G) \leq n^\epsilon$ and
    a decomposition tree $T_r(G,\pi,\ell)$ tales almost linear time by
    \Cref{lem:wcol-approx} and \Cref{lem:tree-good}.

    Using Algorithm \ref{alg:dynprog} on $G$, $T$, $\Omega_\xi$ and $\xi$ for
    every complete conjunctive clause $\xi$ and taking the best result of those
    calls, gives us by \Cref{lem:algo-correct} the correct result for the stated
    problem.

  The (non-recursive) computation of a child takes $t$ almost linear time in
  $V(G_t)$. Also, for every child $t'$ of $t$, there is a recursive call. We get
  the following recurrence relation $R(j,n)$ for the time needed to evaluate a
  node $t$ at level $j$ and $n = |G_t|$:
  \begin{align*}
    R(0, n) &\leq c \\
    R(j,n) &\leq \sum_{X \in \cal X_r} c R(j-1, |X|) + c n^{1+\delta} \qquad \text{for all $j \geq 1$}
  \end{align*}
  In \cite{GroheKS17}, the authors showed that $R(j,n)$ can be bounded by
  $c^\ell n^{1+\epsilon}$. As $c$ and $\ell$ only depend on $\phi, \cal C$ and
  $\epsilon$, we get the desired result.
\end{proof}

\fi

\section{Characterizing Almost Nowhere Dense Graph Classes}
\label{sec:almost-nd}

In this section, we provide various characterizations
of almost nowhere dense classes, i.a.\ via bounded depth minors
and generalized coloring numbers.

\begin{definition}[Almost nowhere dense]
A graph class $\cal C$ is \emph{almost nowhere dense} if for every $r \in
\N$, $\varepsilon > 0$ there exists $n_0$ such that no graph $G\in\cal
C$ with $|G| \ge n_0$ contains $K_{\lceil |G|^\varepsilon \rceil}$
as a depth-$r$ minor.
\end{definition}

\begin{theorem}\label{thm:characterization_almost_nowhere_dense}
    Let $\cal C$ be a graph class.
    The following statements are equivalent.
    \begin{enumerate}
        \item $\cal C$ is almost nowhere dense.
       \item For every $r \in \N$, $\varepsilon > 0$ there exists $n_0$
            such that no graph $G \in \cal C$ with $|G| \ge n_0$
            contains $K_{\lceil |G|^\varepsilon \rceil}$ as a depth-$r$ minor.
        \item For every $r \in \N$, $\varepsilon > 0$ there exists $n_0$
            such that no graph $G \in \cal C$ with $|G| \ge n_0$
            contains $K_{\lceil |G|^\varepsilon \rceil}$ as a depth-$r$ topological minor.
        \item \label{subdiv_anwd_char}For every $r \in \N$, $\varepsilon > 0$ there exists $n_0$
            such that no graph $G \in \cal C$ with $|G| \ge n_0$
            contains an $r'$-subdivision of $K_{\lceil |G|^\varepsilon
            \rceil}$ as a subgraph for any $r' \le r$.
        \raggedright
        \item For every $r \in \N$, $\varepsilon > 0$ there exists $n_0$
            such that
            $\wcol_r(G) \le |G|^\varepsilon$
            for every graph $G \in \cal C$ with $|G| \ge n_0$.
        \item \label{col_anwd_char}For every $r \in \N$, $\varepsilon > 0$ there exists $n_0$
            such that
            $\col_r(G) \le |G|^\varepsilon$
            for every graph $G \in \cal C$ with $|G| \ge n_0$.
    \end{enumerate}
\end{theorem}

The characterizations from \Cref{thm:characterization_almost_nowhere_dense} are very similar to those for nowhere dense classes.
The only difference in the characterizations 1.\ to 4.\ would be the size of the forbidden cliques: for nowhere dense classes, the size would be $f(r)$ instead of $\lceil |G|^\varepsilon \rceil$.
Similarly, if we would substitute ``for every $G \in \cal C$'' with ``for every subgraph $G \subseteq H \in \cal C$'' in characterizations 5 and 6 would characterize nowhere dense classes. 
Note that every almost nowhere dense class which is monotone, i.e. closed under taking subgraphs, is also nowhere dense. 

Conversely, if a class $\cal C$ is almost nowhere dense, then its subgraph-closure $\cal C_\subseteq$ is not almost nowhere dense in general. 
Consider for this the class of graphs which for every $n \in \N$ contains independent set of size $n$ with a clique of size $\log n$, i.e. the graph $I_n \cup K_{\log n}$. This class is almost nowhere dense but its subgraph-closure contains cliques $K_n$ of every size $n$ as member, and so, all graphs.

\iflongpaper


We need the following theorem by Grohe, Kreutzer and
Siebertz~\cite{GroheKS2013}, which in turn builds upon the original results of
Kierstead and Yang~\cite{coloringdefinition} and Zhu~\cite{Zhu2009}.


\begin{proposition}[{\cite[Theorem 3.3]{GroheKS2013}}]
    There exists a function $f \colon \N \to \N$
    such that for all $d,r \in \N$ and all classes $\cal C$ of graphs, if the class
    of all topological depth-$r$ minors of $\cal C$ is $d$-degenerate then $\col_r(G') \le f(r)\cdot d$ for every subgraph $G' \subseteq G$ of a graph $G \in \cal C$.
\end{proposition}

By setting $\cal C$ to be the class containing only a single graph
and reversing the statement, we get the following statement that better suits our needs.

\begin{corollary}\label{cor:deg}
    There exists a function $f \colon  \N \to \N$
    such that for all $d,r \in \N$ and all graphs $G$,
    if $\col_r(G) \ge f(r) \cdot d$ 
    then there exists a topological depth-$r$ minor of $G$ that is not $d$-degenerate.
\end{corollary}

We also need the following observations about degeneracy.

\begin{fact}\label{fact:core}
    A $d$-core of a graph $G$ is a maximally connected sub\-graph of $G$ in which all vertices have degree at least $d$.
    The degeneracy of a graph $G$ is the largest number $d$ for which $G$ has a $d$-core. 
    If a graph is not $d$-degenerate then it has a $(d+1)$-core and therefore a subgraph in which all vertices have degree at least $d+1$.
    A graph with degeneracy $d$ has at most $nd$ edges.
\end{fact}

Combining \Cref{cor:deg} and \Cref{fact:core} yields:

\begin{lemma}\label{lem:high_col_implies_subdiv}
    Let $\rho >0$ and $r \in \N$. There exists $n_0 = n_0(\rho,r)$ and $\mu = \mu(\rho) > 0$ such that all
    graphs $G$ on $n \ge n_0$ vertices with $\col_r(G) \ge n^{1/\rho}$ contain a
    ${\le} (r+1) (9^{2\rho} + 1)$-subdivision of $K_{\lceil n^\mu \rceil}$ as a subgraph.
\end{lemma}

\begin{proof}
    Let $G \in \cal C$ be a graph of order $n$ with $\col_r(G) \ge n^{1/\rho}$
    and $f(r)$ be the function of \Cref{cor:deg}.
    By \Cref{cor:deg}, $G$ has a topological depth-$r$ minor that is not $n^{1/\rho} / f(r)$-degenerate.
    According to \Cref{fact:core}, $G$ also has a topological depth-$r$ minor $H$ in which all vertices have degree at least $n^{1/\rho} / f(r)$.
    Without loss of generality, we can assume $n \ge n_0(\rho,r)$ to be large enough that
    $n^{1/\rho} / f(r) \ge n^{1/2\rho}$.
    Then \Cref{prop:212} guarantees that there exists $\mu(\rho)$ such that $H$ has a $\le 9^{2\rho}$-subdivision of $K_{\lceil n^{\mu(\rho)}\rceil}$ as a subgraph.
    By transitivity, this means $G$ has a ${\le} (r+1) (9^{2\rho} + 1)$-subdivision of $K_{\lceil n^{\mu(\rho)}\rceil}$ as a subgraph.
 
\end{proof}

We use the following statements about subdivided cliques in graphs with polynomial minimum degree.

\begin{proposition}[{\cite[Lemma 2.12]{GroheKS2013}}]\label{prop:212}
    Let $\rho > 1$. There exists $n_0 = n_0(\rho)$ and $\mu = \mu(\rho) > 0$ such that all
    graphs $G$ on $n \ge n_0$ vertices with minimum degree at least $n^{1/\rho}$ contain a $9^\rho$-subdivision of $K_{\lceil n^\mu \rceil}$ as a subgraph.
\end{proposition}

\begin{proposition}[{\cite[Lemma 3.14]{DvorakThesis}}]\label{prop:one_sub}
There exists $n_0$ such that every graph $G$ with $n \ge n_0$ vertices
and minimum degree at least $4n^{0.6}$ contains a $1$-subdivision of $K_{\lceil n^{0.1}\rceil}$ as a subgraph.
\end{proposition}

Combining these two statements yields the following useful observation.
\begin{lemma}\label{lem:exact_subdiv}
    There exists $n_0$ such that every graph $G$ containing an ${\le}r$-subdivision of $K_n$ as a subgraph
    with $n \ge r\cdot n_0$ also contains an $r'$-subdivision of $K_{\lceil n/r^{0.05} \rceil}$ as a subgraph for some $r' \le 2r+1$.
\end{lemma}

\begin{proof}
    Let $n'_0$ be the constant from \Cref{prop:one_sub} and $n_0 \ge 4 n'_0+4$.
    Let $G$ be a ${\le}r$-subdivision of $K_n$ with $n \ge r \cdot n_0$.
    Then there exists $r' \le r$, such that at least $n(n-1)/2(r+1)$ edges of $K_n$ are subdivided exactly $r'$-times in $G$.
    This yields a graph $H$ with $n$ vertices and $n(n-1)/2(r+1)$ edges
    such that $G$ contains an $r'$-subdivision of $H$ as a subgraph.
    A graph with degeneracy $d$ has at most $nd$ edges.
    Thus, by \Cref{fact:core}, $H$ has a $(n-1)/2(r+1)$-core, i.e., a subgraph with minimum degree at least $(n-1)/2(r+1)$.
    Since $n \ge r\cdot n_0 \ge 4r n'_0+4r$, we have that $(n-1)/2(r+1) \ge n'_0$.
    Then by \Cref{prop:one_sub}, $H$ contains a $1$-subdivision of a complete graph of order $\lceil ((n-1)/2(r+1))^{0.1} \rceil$ as a subgraph.
    Since $G$ contains an $r'$-subdivision of $H$ as a subgraph, this means that
    $G$ contains a $2r'+1$-subdivision of order $\lceil ((n-1)/2(r+1))^{0.1} \rceil$ as a subgraph.
    Without loss of generality, we can assume $n \ge n_0$ to be large enough that
    $\lceil ((n-1)/2(r+1))^{0.1} \rceil \ge \lceil n/r^{0.05} \rceil$.
\end{proof}

At last, we use all these observations to obtain a characterization of almost nowhere dense classes.
\begin{proof}[Proof of \Cref{thm:characterization_almost_nowhere_dense}] 
    The equivalence 1.\ $\Leftrightarrow$ 2.\ is by definition.
    For convenience, we show equivalence of the inverse of the remaining statements.
    Let us prove $\neg$2.\ $\Rightarrow$ $\neg$3.
    Let $\omega_r(G)$ (or $\tilde \omega_r(G)$) be the largest value of $t$ such that $G$ has $K_t$ as depth-$r$ minor (or depth-$r$ topological minor).
    According to \cite[Corollary 2.20]{notes},
    \begin{equation}\label{eq:abc}
    \tilde\omega_r(G) \le \omega_r(G) \le 1+(\tilde\omega_{10r}(G)+1)^{10r}.
    \end{equation}
    If $\neg$2.\ holds then there exists $r$, $\varepsilon > 0$ and an infinite sequence of graphs $G_1,G_2,\dots$ such that
    $\omega_r(G_i) \ge |G_i|^\varepsilon$.
    Then there also exists $\varepsilon' > 0$ and $c$ such that
    $\omega_r(G_i) \ge 1 + (|G_i|^{\varepsilon'}+1)^{10r}$
    for all $i \ge c$.
    By (\ref{eq:abc}), 
    $\tilde\omega_{10r}(G_i) \ge |G_i|^{\varepsilon'}$
    for $i \ge c$, which implies~$\neg$3.

    The implication $\neg$3.\ $\Rightarrow$ $\neg$4. follows from \Cref{lem:exact_subdiv}.

    The implication $\neg$4.\ $\Rightarrow$ $\neg$2. holds, since every graph that contains an $r$-subdivision of a graph $H$ as a subgraph
    also contains $H$ as depth-$r$ minor.

    Furthermore, $\neg$5.\ $\Leftrightarrow$ $\neg$6., since 
    $\col_r(G) \leq \wcol_r(G) \leq \col_r(G)^r$ for every graph $G$.

    Next, $\neg$6.\ $\Rightarrow$ $\neg$3. follows from \Cref{lem:high_col_implies_subdiv}.

    At last, we prove $\neg$3.\ $\Rightarrow$ $\neg$4.:
    Assume a graph $G$ contains $K_t$ as a depth-$r$ topological minor, where the principal vertices are $P \subseteq V(G)$, $|P|=k$.
    Let $\pi$ be an ordering of $G$ and $v \in P$ be maximal with respect to $\pi$.
    For every $w \in P$ let $m_w$ be the smallest vertex with respect to $\pi$ on the path from $v$ to $w$ in the depth-$r$ topological minor model.
    Then $\{m_w \colon w \in P\} \subseteq \WReach_{r+1}[G,\pi,v]$. Since all paths from $v$ to $P$ share no vertex except $v$, the vertices $m_w$ are all distinct. This means $\wcol_{r+1}(G) \ge t$.
\end{proof}

\fi

\section{Approximation on Almost Nowhere Dense}\label{sec:approx}

In this section we consider the same problem as before, 
i.e., finding vertices for $\bx$ that satisfy $\cnt y \phi(\bx y) > N$
but on almost nowhere dense classes of graphs. Here, we give an approximation
algorithm with an additive error. For this, we use completely different
techniques compared to \Cref{sec:exact}. We first show how to reduce the
corresponding model-checking problem to approximate sums over unary functions.
\iflongpaper
The procedure from \Cref{lem:res} is the only source of error.
\fi
Then we present the approximate optimization algorithm in \Cref{thm:opti}.

The main result of this section is the following approximate optimization
algorithm with additive error.

\begin{restatable}{theorem}{theoremopti}
    \label{thm:opti}
    There exists a computable function $f$ such that for every graph $G$
    and every quantifier-free first-order formula $\phi(y\bx)$
    we can compute
    a vertex tuple $\bu^*$
    with \[
        |\max_{\bu} [[\cnt y \phi(y \bu)]]^G - [[\cnt y \phi(y \bu^*)]]^G|\leq 4^{|\phi|} \wcol_2(G)^{O(|\phi|)}
    \]
    in time $\wcol_{f(|\phi|)}(G)^{f(|\phi|)}n$.
\end{restatable}

For the \emph{approximate model-checking} problem with an additive error $\delta$, similar to \cite{DreierR2021}, we want an algorithm such that 
\begin{enumerate}
    \item the algorithm returns ``yes'' only if $G$ satisfies the formula,
    \item returns ``no'' only if $G$ does not satisfy the formula,
    \item returns $\bot$ only if the optimum is within $\delta$ to $N$.
\end{enumerate}
The option $\bot$ can be seen as ``I do not know'' as the computed result and the desired result are so close that the difference falls into the additive error $\delta$.

Given the approximate optimization algorithm from \Cref{thm:opti}, we can easily
build an approximate model-checking algorithm as described above for the formula
$\exists \bx \cnt y \phi(y \bx) > N$ by computing a vertex tuple $\bu^*$ from
the theorem. If $N - [[\cnt y \phi(y \bu^*)]]^G  \leq \delta$, answer~$\bot$.
Otherwise, answer ``yes'' or ``no'' according whether $[[\cnt y \phi(y
\bu^*)]]^G > N$ or not. Note that $\delta$ cannot be chosen freely as it depends on the graph (respectively, its weak coloring numbers).

The runtime of the algorithm from \Cref{thm:opti} is fpt if the weak
$r$-coloring numbers are bounded by $n^\varepsilon$ for $r \leq f(|\phi|)$. This
is the case for almost nowhere dense classes. This is in contrast to the results
of \cite{DreierR2021} where the running time of their algorithms is bounded by
$f(\wcol_{f(|\phi|)})||G||$ which is fpt on classes of bounded expansion but is
not fpt on nowhere dense and almost nowhere dense classes.

This gives us the following corollary.
\corollaryapproxanwd*

\iflongpaper

\subsection{Functional structures}
We will heavily rely on functional representations of directed graphs, where
edges are replaced with functions mapping the endpoint of an edge to its
startpoint. For this, we need a functional signature consisting of
$\Delta^-(\vG)$ many unary functions. They were used by Durand and
Grandjean~\cite{DurandG2007}, as well as Kazana and Segoufin~\cite{KazanaS2013},
and also by Dvořák, Kráľ, and Thomas~\cite{DvorakKT2013}. A big advantage of
functional structures is that short paths can be expressed in a quantifier-free
way.

We focus on representations of graphs where the arcs are all directed along a
fixed vertex ordering. One can imagine that this vertex ordering witnesses that
the weak coloring number of the graph is small, which means that the number of
symbols in the signature is also small.

\begin{definition} \label{def:fct-struct}
For a graph $G$ with vertex ordering $\pi$ we define 
the \emph{functional representation} $\fG$ of $G$ w.r.t.\ $\pi$ as
a functional structure with universe $V(G)$ and functional signature 
$(f_1,\ldots,f_t)$ for $t=\wcol_1(G, \pi)$ where $[[f_i(u)]]^{\fG} = v$ 
if $v$ is the $i$th weakly 1-reachable vertex from $u$ with regard to $\pi$ 
and $[[f_i(u)]]^{\fG} = u$ if $i > |\WReach_1[G, \pi, u]|$.
\end{definition}
We denote $G$ as the \emph{underlying graph} of $\fG$ and $\fG$ has
multiplicity~$1$.

\begin{figure}
\centerline{
            \includegraphics{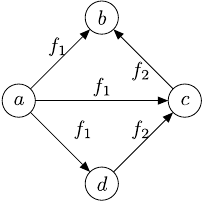}
            \vbox{%
             \hbox{\quad{\vbox{\halign{$#$\quad&$#$\cr
              f_1(a)=a & f_2(a)=a\cr
              f_1(b)=a & f_2(b)=c\cr
              f_1(c)=a & f_2(c)=d\cr
              f_1(d)=a & f_2(d)=d\cr
             }}}}%
             \vskip2\baselineskip
             \hbox{{\includegraphics{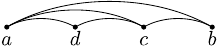}}}
            }
}
\caption{A graph and one of its functional representations.}
\label{fig:aug}
\end{figure}

In the following we will define an augmentation that is a special case
of a transitive fraternal augmentation as defined by Nešetřil and
Ossona de~Mendez~\cite{sparsity}.  A transitive fraternal augmentation
adds transitive and fraternal arcs to a directed graph.  If $uv$ and
$vw$ are arcs, then $uw$ is the corresponding transitive arc and if
$uv$ and $uw$ are arcs then $vw$ and $wv$ are fraternal arcs.  While
a transitive arc is unique there are two possible fraternal arcs and
in a transitive fraternal augmentation only one of them has to be
added.  In our case the direction of the fraternal arcs will be
determined by an order~$\pi$, which allows us to bound the indegree of
the resulting directed graph by a weak coloring number.

\begin{definition} \label{def:augmentation}
For a graph $G$ and a vertex ordering $\pi$ we define the
\emph{augmentation} $\hG$ of $G$ as the expansion of $\fG$ with
the functional symbols $h_i$ for $i \in [\wcol_2(G, \pi)]$ where
$[[h_i(u)]]^{\hG} = v$  if $v$ is the $i$th weakly 2-reachable vertex
from~$u$ w.r.t.\ $\pi$ and $[[h_i(u)]]^{\hG} = u$ if $i >
|\WReach_2[G, \pi, u]|$.
\end{definition}
Note that the underlying graph of $\hG$ is not $G$.  However,
$\wcol_r(\hG,\pi) \leq \wcol_{2r}(\fG, \pi)$ for every $r \in [|G|]$.
Moreover, the multiplicity of $\fG$ is $1$ and that of $\hG$ is~$2$.
Let us look again at the graph~$\fG$ in Figure~\ref{fig:aug}.  If
we construct $\fG^2$ then $h_1(b)=a$, $h_2(b)=b$, and $h_3(b)=c$
because $a$, $b$, and $c$ are weakly $2$-reachable from $d$.
The augmentation has three new functional symbols because
$\wcol_2(G,\pi)=2$.  We have $h_3(b)=f_2(b)=c$, which shows that the
multiplicity is~2.

\begin{lemma} \label{lem:construct-fct}
Given a graph $G$ with an ordering $\pi$, one can compute $\fG$ and $\hG$
in time $O(||G||)$ and $O(\wcol_2(G,\pi)^2||G||)$, respectively.
\end{lemma}
\begin{proof}
    Computing $\fG$ is easy. For this, the edges of $G$ have to be directed
    from right to left w.r.t.\ to~$\pi$.  The left neighbors of each vertex
    $v$ have to be assigned to $f_i(v)$ in order.  We only store $f_i(v)$
    for the $|\WReach_1(G, \pi, v)|$ many left neighbors of~$v$.
    
    We derive $\hG$ from $\fG$ and store it as follows:  
    As $\fG$ and $\hG$ agree except for the new
    functions $h_i$, we need to store only the latter.  For every $u\in
    V(G)$ we store all $h_i(u)$ with $h_i(u)\neq u$, i.e., only the first
    $|\WReach_2(u)|\leq d(u)\wcol_1(G)$ ones, by going through all neighbors
    $x$ of~$u$ and then to all neighbors of $x$ on the left of~$u$.
    All found vertices are then sorted in $O(\wcol_2(G,\pi)\log\wcol_2(G,\pi))$
    time.
    Altogether $\hG$ can be computed in time
    \begin{equation*}
    \smash{\sum_{u\in V(G)}} O\bigl(d(u)\wcol_1(G,\pi)+
    \wcol_2(G,\pi)\log(\wcol_2(G,\pi))\bigr)
    =O\bigl(||G||\wcol_2(G, \pi)^2\bigr).
    \end{equation*}
    Note that we can compute $h_i(u)$ in constant time if we have this
    representation of~$\hG$.
\end{proof}

\subsection{Multiplicity}

\begin{definition}
A structure $G$ with signature $\sigma$ has \emph{multiplicity} $m$ if
for every distinct pair $u, v \in V(G)$, the number of function symbols
$f \in \sigma$ with $u = f^G(v)$ or $v = f^G(u)$ and relation symbols $R
\in \sigma$ such that $R^G(u,v)$ is at most $m$.

A quantifier-free conjunctive clause $\omega(\bx)\in \FO[1, \sigma]$
has multiplicity $m$ if for distinct $i, j \in [|\bx|]$ there are at
most $m$ positive literals of the form $f(x_i) = x_j$ or vice versa.
\end{definition}

\begin{definition}
We define $\FO[1, \sigma, 2]$ to be all quantifier-free conjunctive
clauses with at most one function application, signature $\sigma$ and
multiplicity at most 2 where for $i, j$ the literal $f(x_i) =
x_j$ may appear only if $i \neq j$ and there is no literal $x_i = x_j$.

A clause $\omega(\bx)\in\FO[1,\sigma,2]$ is \emph{complete} if for
every $i\neq j$ and every $f\in\sigma$ either
$x_i=x_j$ or $x_i\neq x_j$ is contained in~$\omega$.
Furthermore, if $x_i\neq x_j\in\omega$ then for every function symbol
$f\sigma$ either
$f(x_i)=x_j$ or
$f(x_i)\neq x_j$ must be contained in~$\omega$.  If, on the other
hand, $x_i=x_j\in\omega$, then no other literal containing $x_i$ and
$x_j$ is allowed in~$\omega$.
\end{definition}
The complicated interaction between literals of the forms $x_i=x_j$
and $f(x_i)=x_j$ stems from the absence of self-loops.




\subsection{Decomposing Formulas into Simpler Ones}

Let $\omega(y \bx)$ be a conjunctive clause in a functional signature $\sigma$.
Over the course of this section, we will repeatedly decompose such clauses into
three conjunctive clauses $\tau(y), \psi(\bx)$, $\Delta^=(y \bx)$ such that
$\omega(y\bx) \equiv \tau(y) \land \psi(\bx) \land \Delta^=(y \bx)$. We will
always require that
\begin{itemize}
    \item $\tau(y)$ has functional depth at most two and contains only
    the variable~$y$,
    \item $\psi(\bx)$ contains literals of the form $x_p = f_j(x_q)$,
    and
    \item $\Delta^{=}(y\bx)$ contains only literals of the form $f(y)=x_i$, $f(x_i)=y$ and $y=x_i$ for $f \in \sigma$.
\end{itemize}
Formulas with the names $\tau(y), \psi(\bx)$, $\Delta^=(y \bx)$
will always be conjunctive clauses with the properties above and
will refer to a decomposition of a clause $\omega$, even if this is not explicitly mentioned.
We call $\Delta^=(y \bx)$ also the \emph{positive mixed literals} of~$\omega$.


The first step of the algorithm is to decompose a relational
quantifier-free formula~$\phi$ into a set of weighted conjunctive
clauses with a restricted form.
Also, we switch from a relational representation of the graph and the
formula to a functional representation.
The form of the clauses will be simple in the sense that 
there is only one literal that contains both $y$ and a variable from
$\bx$, which will allow us to use \Cref{lem:l2}.
We will use the notion of multiplicity throughout this series of
lemmas to be able to apply \Cref{lem:restrict_signature}.

The approximative error occurs in the procedure of \Cref{lem:res}.
Here, clauses with literal $y=f(x)$ are ignored as these cannot be handled with
our techniques. However, their impact on the evaluation is relatively small as
the vertices described by this literal have to be in $\WReach_r[G,\pi,v]$ for
some small number of vertices $v \in V(G)$. Also, the number of clauses with
this literal is quite small. 

\begin{lemma}
    \label{lem:res}
    Consider a graph $G$ with order $\pi$ and a quantifier-free first order
    formula $\phi(\bx y)$, both with signature $\sigma$. In time $(\wcol_2(G,
    \pi)+1)^{O(|\phi|)}$, one can construct a set $\Omega$ with the following
    properties:
    
    \begin{enumerate}
    \item The set $\Omega$ contains pairs $(\mu, \omega(\bx y))$ where $\mu \in
    \Z$ and $\omega$ is of the form $\tau(y) \land \psi(\bx) \land f(y) = x_i$ 
    \item $\psi$ has only positive literals with at most one function
    application,
    \item $|\Omega| \leq (\wcol_2(G)+1)^{O(|\phi|)}$,
    \item $|\omega| \leq 2 |\phi| +1$ for each $(\mu,\omega) \in \Omega$,
    \item for all $\bu \in V(G)^{|\bx|}$, and with $\delta \coloneqq 4^{|\phi|}
    \wcol_2(G)^{O(|\phi|)} $\\
    $\displaystyle \sum_{ (\mu, \omega) \in \Omega} \mu [[\cnt y \omega(\bu
    y)]]^{\hG} - \delta \leq [[\cnt y \phi (\bu y)]]^{G} \leq \sum_{ (\mu,
    \omega) \in \Omega} \mu [[\cnt y \omega(\bu y)]]^{\hG} + \delta .
    $
    \label{res6}
    \end{enumerate} 
\end{lemma}

Before we can prove this lemma we consider first
\Cref{lem:remove_uneq,lem:remove_eq}. \Cref{lem:remove_uneq} uses
inclusion-exclusion to get rid of negative literals.
Then \Cref{lem:remove_eq} switches to the functional setting.
A clause that results from \Cref{lem:remove_uneq} is then transformed
into a set of clauses with only one mixed positive literal each.
The considered graph changes from the relational graph~$G$,
to its functional representation~$\fG$ and
then its augmentation~$\hG$.
The multiplicity is bounded during this procedure.

\begin{lemma}\label{lem:remove_eq}
    Let $\fG$ be the functional representation of a graph $G$ 
    and $\omega$ be a conjunctive clause of the form 
    $\omega(y \bx) = \tau(y) \land \psi(\bx) \land \Delta(y\bx)$
    where $\Delta(y\bx)$ is a conjunction of positive literals $f_i(y) = x_p$. 

    Then a set of clauses $\Omega$ with the following properties can be computed
    in time $(\wcol_2(G,\pi)+1)^{O(|\omega|)}$: 
    
    \begin{enumerate}
    \item Each clause in $\Omega$ is of the form $\tau'(y) \land \psi'(\bx)
    \land f(y) = x_i$,
    \label{req1}
    \item $\psi' \in FO[1,\sigma,2,\bx]$ contains only positive literals with at
    most one function application,
    \item $|\omega'| \leq |\omega|$ for each $\omega' \in \Omega$,
    \item $|\Omega| \le (\wcol_2(G, \pi)+1)^{|\omega|}$,
    \item \label{G2} for every $\bu \in V(G)^{|\bx|}$, \\
    $\displaystyle [[\cnt y \omega (y \bu)]]^{\fG} = \sum_{ \omega' \in \Omega}
    [[\cnt y \omega'(y \bu)]]^{\hG}.$
    \end{enumerate} 
\end{lemma}
Note that in \ref{G2}. the interpretation inside the sum is over $\hG$ and
\emph{not} over $\fG$.

\begin{proof}
    Remember that the mixed positive literals of a clause are those literals
    contained in the part $\Delta^{=}(y\bx)$ of its decomposition. We start with
    $\Omega = \{\omega\}$ and describe a procedure that picks a clause $\omega'
    \in \Omega$ with $l>1$ mixed positive literals, removes $\omega'$ and
    replaces it with $\wcol_2(G, \pi)+1$ clauses with at most $l-1$ mixed
    positive literals. Once this procedure cannot be applied any longer, each
    clause has exactly one mixed positive literal and the set $\Omega$
    satisfies~\ref{req1}, i.e., that there is only one mixed positive literal.
    Since initially, $\omega$ has at most $|\omega|$ many mixed positive
    literals, $\Omega$ will have size at most $(\wcol_2(G, \pi)+1)^{|\omega|}$
    upon termination of the procedure.

    Let us pick a clause $\omega' \in \Omega$ and describe the procedure
    mentioned above in detail. There are two literals $x_p = f_i(y)$ and $x_q =
    f_j(y)$ in $\Delta^=(y\bx)$ with $i \leq j$. This implies that $x_p$ is
    weakly 2-reachable from $x_q$.

    Let $\Delta'^=(y\bx)$ be the clause obtained from $\Delta^=(y\bx)$
    by removing the literals $x_p = f_i(y)$ and $x_q = f_j(y)$.
    We remove $\omega'$ from $\Omega$ and add the clause
    \[
    x_p = x_q \land x_p = y \land \tau(y) \land \psi(\bx) \land \Delta'^=(y\bx),
    \]
    as well as for each $k \in [\wcol_2(G, \pi)]$ the clause
    \[ 
     x_p \neq x_q \land x_q = f_j(y) \land x_p = h_k(x_q) \land
     f_i(y)
    = h_k(f_j(y))  \land \tau(y) \land \psi(\bx) \land \Delta'^=(y\bx).
    \] 
    This means, one clause is removed and $1+\wcol_2(G, \pi)$ new clauses are
    added to $\Omega$.

    Remember that $h_i(x_q)$ is the $i$th weakly 2-reachable vertex from $x_q$
    or $x_q$ itself. Thus, if
    there are two distinct $k$ and $k'$
    such that $x_p = h_k(x_q)$ and $x_p = h_{k'}(x_q)$ then $x_p = x_q$. This
    implies that the newly added $1+\wcol_2(G, \pi)$ many clauses are mutually
    exclusive. The equivalence follows by observation: As $x_p$ is 2-reachable
    from $x_q$ there exists a $k \in [\wcol_2(G, \pi)]$ such that $x_p =
    h_k(x_q)$. 
    
    With $x_p = h_k(x_q)$ and some syntactic replacements it follow
    that the literal $x_p = f_i(y)$ 
    is equivalent to $f_i(y) = h_k(f_j(y))$.
    Hence, 
    \[ x_p = f_i(y) \land x_q = f_j(y) \land x_p = h_k(x_q) \] 
    is equivalent to
    \[ f_i(y) = h_k(f_j(y)) \land x_q = f_j(y) \land x_p = h_k(x_q). \qedhere \] 
\end{proof}

The idea of this proof is that the number of mixed literals can be decreased.
Two vertices which are weakly 1-reachable from $y$, are connected by a
functional edge in the augmented graph. With a syntactic trick, this can be
expressed with fewer mixed literals. Note that we use functional representations
to express these distance-2 relationships without needing to resort to
quantifiers.

Finally, we are able to prove \Cref{lem:res} by combining
\Cref{lem:remove_uneq,lem:remove_eq}. The other part of this proof is the
transition from relational to functional representations.

\begin{proof}[Proof of \Cref{lem:res}] 
    First, we apply \Cref{lem:remove_uneq} to $G$ and $\phi$, resulting in the set $\Omega_1$.
    Next, we turn to the functional representation $\fG$ of $G$. The signature
    of $\fG$ is then $\{\,f_i \mid i \in \wcol_1(G,\pi)\,\}$. Let
    $(\mu,\omega_1) \in \Omega_1$. Note that $\omega$ contains only positive
    literals. We construct $\omega_2$ by replacing every adjacency atom $E(a,b)$
    of $\omega_1$ for $a,b \in y\bx$ with
    \begin{equation}\label{eq:disjunction}
    a = b \lor \!\!\!\!\!\!\bigvee_{i \in [\wcol_1(G)]}\!\!\!\!
    (f_i(a) = b \land f_i(a) \neq a) \lor (f_i(b) = a \land f_i(b) \neq b).
    \end{equation}
    Note that the disjunction in (\ref{eq:disjunction}) is mutually exclusive in
    $\fG$. Thus, each adjacency atom gets replaced with a mutually exclusive
    disjunction over at most $2\wcol_1(G)+1$ conjunctive clauses. Therefore,
    transforming $\omega_2$ into disjunctive normal form yields at most
    $(2\wcol_1(G)+1)^{|\phi|}$ many mutually exclusive clauses. We place each of
    those clauses, with a weight $\mu$ into a new set $\Omega_2$.
    This procedure is repeated for all $(\mu,\omega_1) \in \Omega_1$.

    \begin{equation}\label{eq:lambdasum2}
        [[\cnt y \phi(y\bu)]]^G = \smash{\!\!\sum_{(\mu,\omega)\in\Omega_1}}\!\!
        \mu [[\cnt y \omega(y\bu)]]^{G} 
        = \!\!\sum_{(\mu,\omega)\in\Omega_2}\!\!
	\mu [[\cnt y \omega(y\bu)]]^{\fG}  .
    \end{equation}
    Each clause in $\Omega_2$ to is of the form $\tau(y)\land\psi(\bx)\land
    \Delta^=(y\bx)$.
    
    Consider the set $\Upsilon \subseteq \Omega_2$ of weighted clauses which
    contain a (positive or negative) literal of the form $y = f(x)$. We claim
    that $|\sum_{(\mu,\omega)\in\Upsilon} \mu [[\cnt y \omega(y\bu)]]^{\fG}|
    \leq 4^{|\phi|} \wcol_2(G)^{O(|\phi|)}$.

    The size of $\Upsilon$ is bounded by the size of $\Omega_2$ which is
    $\wcol_2(G)^{O(|\phi|)}$. Also, for each $\omega \in \Upsilon$ it holds that
    $[[\cnt y \omega(y \bu)]] \leq \wcol_1(\hG)$ as there is only one choice of
    $y$ for every fixed tuple $ \bu$ (due to the positive literal $y = f(u_i)$).
    As the weights of a clause in $\Omega_1$ is bounded by $4^{|\phi|}$ and the
    disjunction in \ref{eq:disjunction} is mutually exclusive, the weights of
    $\Omega_2$ are also bounded by $4^{|\phi|}$. The claim follows directly.

    Hence, removing $\Upsilon$ from $\Omega_2$ changes its evaluation by at most $4^{|\phi|} \wcol_2(G)^{O(|\phi|)}s$ additively, i.e.,
    \[
    [[\cnt y \phi(y\bu)]]^G
        = \sum_{(\mu,\omega)\in\Omega_2} \mu [[\cnt y
	    \omega(y\bu)]]^{\fG} = 
        \sum_{(\mu,\omega)\in\Omega_2 \setminus \Upsilon} \mu [[\cnt y \omega(y\bu)]]^{\fG} \pm 4^{|\phi|} \wcol_2(G)^{O(|\phi|)} .
    \] 
    We can apply \Cref{lem:remove_eq} to $\Omega_2 \setminus \Upsilon$.
    The resulting set gives us the desired set $\Omega$:
    Condition~\ref{res6} follows from (\ref{eq:lambdasum2}) and \Cref{lem:remove_eq}.
    The size of $\Omega$ is bounded by $|\Omega_2| \cdot (\wcol_2(G, \pi)+1)^{O(|\phi|)}$ which again is bounded by $(\wcol_2(G, \pi)+1)^{O(|\phi|)}$.
    Computing $\Omega$ is dominated by its size.
\end{proof}

\subsection{From Formulas to Weights}

The next lemma breaks down the evaluation of $\cnt y \phi(y\bu)$ into evaluating
quantifier-free first-order clauses on \emph{single} variables with ``weights.''
Note that there is no counting quantifier or dependence on $y$ in these clauses.
The lemma is essentially an adaption from \cite{DreierR2021} (Theorem 3 and
Lemma 6).

\begin{lemma} \label{lem:l2} Consider as input $\hG$ and a set $\Omega'$ of
    conjunctive clauses of the form $\tau(y) \land \psi(\bx) \land f(y) = x_i$,
    where $\psi(\bx)$ is in $FO[1,\sigma,2,\bx]$. In time $f(|\phi|)
    \wcol_{2}(G)^{O(|\phi|)} ||G||$ we can compute a set of conjunctive
    clauses~$\Omega$ with free variables $\bx $, as well as functions
    $c_{\omega, i}(v) \colon V(G) \to \Z$ for $\omega \in \Omega$ and $i \in \{
    1, \dots, |\bx| \}$ such that for every $\bu \in V(G)^{|\bx|}$ there exists
    exactly one formula $\omega \in \Omega$ with $\hG \models\omega(\bu)$ and
    for such a formula $\omega$
    \[
    \sum_{( \mu, \omega') \in \Omega'}[[\cnt y \omega'(y \bu)]]^{\hG}
    = \smash{\sum_{i=1}^{|\bx|}} c_{\omega, i}(u_i).
    \]
    The size of $\Omega$ is $\wcol_2(G, \pi)^{O(|\phi|)}$. The length of each
    $\omega \in \Omega$ is $O(\wcol_2(G, \pi))$, $\omega$ has multiplicity~$2$,
    and each literal in $\omega$ has at most one function application (e.g.,
    $f(x_i) = x_j$).
\end{lemma}

\begin{proof}
    Let $\rho$ be the signature of $\hG$ (namely, $(f_i)_{i \in [\wcol_1(G,
    \pi)]} \cup (h_i)_{i \in [\wcol_2(G, \pi)]}$) and $\Omega \subseteq \FO[1,
    \rho, 2]$ be the set of all complete conjunctive clauses with at most one
    function application per literal, signature $\rho$, multiplicity~2, free
    variables $\bx z$. This set has three important properties: First, for every
    $\bu \in V(G)^{|\bx|}$ there exists exactly one $\omega \in \Omega$ with
    $\hG \models \omega(\bu)$.
    
    Second, for every $\omega \in \Omega$ and conjunctive clause $\psi(\bx) \in
    \FO[1, \rho, 2]$ either $\omega \models \psi$ or $\omega \models \neg \psi$.
    The size of $\Omega$ is bounded by $|\rho|^{2|\bx|^2}$ as each complete
    conjunctive clause in $\FO[1, \rho, 2]$ can be identified with its positive
    literals.
    
    Let now $\bu \in V(G)^{|\bx|}$, $\bigl(\mu, \tau(y) \land \psi(\bx) \land
    g(y) = x_i\bigr) \in \Omega'$ and $\omega \in \Omega$ such that $\hG \models
    \omega(\bu)$. If $\omega \models \neg \psi$ then
    \[[[\#y\, \tau(y) \land \psi(\bu)\allowbreak \land g(y) = x_i]]^{\hG} = 0.\]
    Otherwise, if $\omega \models \psi$ then
    \[\smash{[[\#y\, \tau(y) \land \psi(\bu) \land g(y) = x_i]]^{\hG}} = [[\#y\,
    \tau(y) \allowbreak \land g(y) = x_i\ ]]^{\hG} . \] 
    Using this observation,
    we define for every $\omega \in \Omega$ and $i \in \{1,\dots,|\bx|\}$ a set
    $\Gamma_{\omega,i}$ by iterating over all formulas $\omega \in \Omega$ and
    $(\mu,\tau(y) \land \psi(\bx) \land g(y) = x_i \in \Omega'$ and adding
    $(\mu, \tau(y) \land g(y) = x_i$ to $\Gamma_{\omega,i}$ if $\omega \models
    \psi$.
    Now for every $\bu \in V(G)^{|\bx|}$ there exists exactly one formula
    $\omega \in \Omega$ with $\hG \models \omega(\bu)$, and for such a
    formula~$\omega$
    \begin{equation}
    \sum_{(\mu,\omega')\in\Omega'}\mu [[\#y\,\omega'(y\bu)]]^{\hG} 
    =\sum_{i = 1}^{|\bx|}
    \sum_{(\mu, \tau(y) \land g(y) = x_i \in\Gamma_{\omega,i}}
    \!\!\!\!\mu [[\#y\,\tau(y) \land g(y) = u_i]]^{\hG}.
    \label{eq:y}
    \end{equation}
    
    Fix one set $\Gamma_{\omega,i}$. For every formula $(\mu,\tau(y) \land g(y)
    = u_i) \in \Gamma_{\omega,i}$ we  construct a function $c$ with $c(v) =
    [[\#y\,\tau(y) \land g(y) = u_i]]^{\hG}$ in time $O(|\phi| \cdot ||G||)$ by
    the following algorithm. Note that $\tau$ is quantifier-free and its size is
    bounded by $O(|\phi|)$.

    \begin{algorithm}[H]
    \For{$u\in V(\hG)$ with $\hG\models\tau(u)$}{
    $c(g(u)) \gets c(g(u))+1$}
    \end{algorithm}
    
    Let $c_{\omega,i}$ be the sum over all such functions $c$
    for formulas in $\Gamma_{\omega,i}$.
    Then
    \begin{equation}\label{eq:z}
    c_{\omega,i}(u_i) =
    \!\!\!\sum_{(\mu, \tau(y) g(y) = u_i \in\Gamma_{\omega,i}}\!\!\!\!
    \mu [[\#y\,\tau(y) \land g(y) = u_i]]^{\hG}.
    \end{equation}
    Combining Equations (\ref{eq:y}) and (\ref{eq:z})
    yields our statement.
    
    Computing $\Omega$ takes linear time. Computing all $\Gamma_{\omega,i}$
    takes $O(|\Omega| \cdot |\bx| \cdot |\Omega'|)$ time. Computing each
    $c_{\omega,i}$ takes $O(|\Omega'| \cdot |\phi| \cdot ||\hG||)$ time. This
    gives us in total the desired running time.
\end{proof}

For now, the size of the signature of the relational structure and the size
of the clauses depend on $\wcol_2(G,\pi)$ which is too large for our
application.
The following lemma decreases both sizes to a number 
that only depends on $|\bx|$, the number of free variables of the clause.
The underlying structure remains untouched. 

This lemma will be essential to make the running time fpt on nowhere
and almost nowhere dense graph classes, but is not needed for graph
classes of bounded expansion.  Because of the bounded multiplicity the
number of positive literals in $\omega$ is bounded by a function of~$k$
and most literals are negative.  As $\omega$ is complete, we will be
able to treat these negative literals equivalently, when necessary.
We consider the following lemma together with the notion of multiplicity the 
main difference between the approach in this section and the one of 
\cite{DreierR2021}.
\begin{lemma}\label{lem:restrict_signature}
    Let $G$ be a relational structure with signature $\sigma$ with binary,
    symmetric and irreflexive relations and multiplicity~$2$ and $\omega$ a
    complete conjunctive clause in $\FO[1, \sigma, 2]$ over free variables $\bx
    = x_1 \dots x_k$. Then we can compute a relational structure $G'$ and a
    relational clause $\omega'$ both with signature $\rho$ such that $td(G')
    \leq td(G)$, $|\rho| \leq 2k^2$ and $|\omega'| \leq 2^{2{k^4}} + k^2$ in
    time $||G|| + |\omega'|$ and for all $\bu \in V(G)^k$
    \[
    [[\omega(\bu)]]^{G} = [[\omega'(\bu)]]^{G'} .
    \]
\end{lemma}

\begin{proof}
    Let $L^+$ be the positive literals of $\omega$ with two distinct variables
    (e.g., $E(x_i, x_j)$ and not $E(x_i, x_i)$), $\sigma^+$ be the set of
    relational symbols occurring in $L^+$ and $\sigma^- := \sigma \setminus
    \sigma^+$. As the multiplicity of $G$ is at most 2, $|\sigma^+| \leq 2k^2$.
    
    We define a relational structure $G'$ on vertices $V(G)$ where for each $E'
    \in \sigma^+$ the edge relation $E'$ is preserved. Additionally, we
    introduce the edge relation $E^- \coloneqq \bigcup_{E' \in \sigma^-}E'$.
    Note that the underlying graph of $G$ and $G'$ is the same. Hence, their
    tree-depths are identical.
    
    Moreover, the literals of the clause $\omega$ can be partitioned into sets
    $\omega_{\sigma^+}$ and $\omega_{\sigma^-}$ such that $\omega(\bx) \equiv
    \omega_{\sigma^+}(\bx) \land \omega_{\sigma^-}(\bx)$ where
    $\omega_{\sigma^+}$ is the conjunction of literals of $\omega$ with
    (positive and negative) edge relations from $\sigma^+$ and equality and
    $\omega_{\sigma^-}$ the conjunction of literals using negative edge
    relations from~$\sigma^-$.
    
    %
    Also note that $|\omega_{\sigma^+}| \leq 2^{2{k^4}}$ as there are at most
    $k^2$ pairs $x_i$ and $x_j$ and at most $2k^2$ choices for $f \in \sigma_+$.
    Also $\omega_{\sigma^-}(\bx) \equiv \bigwedge_{f \in \sigma^-}
    \bigwedge_{i,j \in [k]} \neg E^-(x_i, x_j)$.
    
    We define a new complete, conjunctive clause $\omega'$ with signature $\rho
    \coloneqq \sigma^+ \cup \{E^-\}$ 
    \[
    \omega'(\bx) \coloneqq \omega_{\sigma^+}(\bu)
    \land \bigwedge_{i, j \in [k]} \neg E^-(x_i, x_j) .
    \]
    
    It is easy to see that for each $\bu \in V(G)^k$ that
    \begin{equation} \label{eq:opti}
    G \models \omega(\bu) \iff G' \models \omega'(\bu) .
    \end{equation}
    The formula $\omega'$ has length at most $2^{2{k^4}}+ k^2$ and that $\omega'
    \in \FO[1, \rho, 2]$
\end{proof}

We are now able to prove our main result of this section. The proof idea works
as follows: Using \Cref{lem:res,lem:l2} we can break down the counting formula
into a sum of vertex weights that depend only on single vertices. Using low
treedepth colorings and an optimization variant of Courcelle's theorem, we can
optimize it in fpt time.

However, this approach is not yet possible as both the signature and the length
of the clauses are not bounded by a function of $|\phi|$, but they depend on the
weak coloring number of $G$. Hence, it is not suited as an input for Courcelle's
theorem. We solve this problem by applying \Cref{lem:restrict_signature} which
yields a shorter, equivalent formula of size $f(k)$.


\begin{proof}[Proof of \Cref{thm:opti}] 
    We use \Cref{lem:wcol-approx} to compute a vertex ordering $\pi$ with
    $\wcol_r(G,\pi)\leq \wcol_{g(r)}(G)^{g(r)}$ for every $r\in\N$ in linear
    time for a computable function $g$. Note that if have a running time or some
    structure of size bounded by $\wcol_{h(r)}(G,\pi)^{h(r)}$ for some
    computable function $h$, then it is also bounded by $\wcol_{f(r)}(G)^{f(r)}$
    for some computable function $f$. This bound is good enough for most of our
    cases.
    
    We use \Cref{lem:construct-fct} to construct $\fG$ and $\hG$ \Cref{lem:res}
    to construct a set $\Omega'$ such that for every $\bu \in V(G)^{ |\bx|}$
    \begin{equation} \label{eq:a}
        \sum_{ (\mu, \omega) \in \Omega'} \mu [[\cnt y \omega(\bu
        y)]]^{\hG} - \delta
        \leq
        [[\cnt y \phi (\bu y)]]^{G} \leq \sum_{ (\mu, \omega) \in \Omega'} \mu [[\cnt y \omega(\bu y)]]^{\hG} + \delta 
    \end{equation}
    where $\delta \coloneqq 4^{|\phi|} \wcol_2(G)^{O(|\phi|)} $. 
    
    Applying \Cref{lem:l2} to $\Omega'$ gives us 
    a set $\Omega$,
    and functions $c_{\omega,i}(v)$ with $c_{\omega,i}(v) = O(|\hG|)$
    such that for every $\bu \in V(G)^{|\bx|}$ 
    \[ 
    \sum_{ (\mu, \omega) \in \Omega'} \mu [[\cnt y \omega(\bu y)]]^{\hG}
    =\sum_{i=1}^{|\bx|}c_{\omega,i}(u_i),
    \] 
    where $\omega \in \Omega$ is the formula with
    $\hG \models \omega(\bu)$.
    Assume for now that we can compute for a given $\omega \in \Omega$
    a tuple $\bu^* \in V(G)^{|\bx|}$ such that
    \begin{equation}\label{eq:newopti}
    \sum_{i=1}^{|\bx|}c_{\omega,i}(u^*_i) = 
    \max_{\bu}\Bigl\{\,\sum_{i=1}^{|\bx|}c_{\omega,i}(u_i) \Bigm|
    \hG \models \omega(\bu)\,\Bigr\}.
    \end{equation}
    Then we could cycle through all $\omega \in \Omega$,
    compute a solution $\bu^*$ satisfying (\ref{eq:newopti}),
    and return the optimal  $\bu^*$ among all of them.
    This gives us a solution to our original optimization problem 
    up to an additive error of $\delta$ resulting from \Cref{eq:a}.
    Thus, from now on, we will concentrate on one formula $\omega \in \Omega$
    and solve the optimization problem~(\ref{eq:newopti}).
    
    It will now be easier for us to work with relational instead of functional
    structures.
    We transform $\hG$ into a relational undirected structure $G'$
    with the same universe via standard methods:
    The unary relations are preserved.
    Additionally, for every function symbol $f$ we add the relation symbol
    $E_f$
    with
    $E^{G'}_f = \{\,(v,f_{\hG}(v)), (f_{\hG}(v), v) \mid v \in V(\vG'),
    f_{\hG}(v) \neq v \,\}$, a symmetric and irreflexive binary relation.
    
    The resulting structure is isomorphic to an undirected graph with both
    vertex- and edge-labels, without self-loops\footnote{We disregard
    self-loops, as they do not contain any additional information.}, and has the
    same weak coloring numbers as $\hG$. We further construct a relational
    conjunctive clause $\omega'(\bx)$ such that $\hG \models \omega(\bu)$ iff
    $G' \models \omega'(\bu)$ for every $\bu \in V(G)^{|\bx|}$. This can be done
    by replacing each literal $f(x_i) = x_j$ by $E_f(x_i, x_j)$.
    Note that this preserves the multiplicity of the clauses, i.e., $\omega'$
    has multiplicity $2$.
    
    As $\wcol_r(G') \leq \wcol_{2r}(G)$ for
    every~$r$~(see~\Cref{def:augmentation}) we can compute a low tree-depth
    coloring with few colors by \Cref{prop:treedepth-col}, i.e., an
    $r$-treedepth coloring with at most $\chi \leq \wcol_{2^{r-2}}(G') \leq
    \wcol_{2^{r-1}}(G)$ colors. Remember that we fixed an $\omega'$ such that
    $\hG \models \omega'(\bu)$. As the tuple $\bu$ is contained in some subgraph
    induced by at most $|\bx|$ colors, we can from now on assume that we are
    working with graphs of bounded tree-depth: Define $\mathcal H$ as the set of
    graphs induced by at most $|\bx|$ colors in $G$. The size of $\mathcal H$ is
    bounded by ${\chi \choose |\bx|} \leq \wcol_{2^{|\bx|-1}}(G, \pi)^{|\bx|}$.
    
    For every $\bu \in V(G)^{|\bx|}$ with $\hG \models \omega'(\bu)$ there
    exists $H \in \cal H$ such that $\bu \in V(H)^{|\bx|}$ and $H \models
    \omega'(\bu)$. In order to optimize (\ref{eq:newopti}), it is therefore
    sufficient to consider every graph $H \in \cal H$ and compute $\bu^* \in
    V(H)^{|\bx|}$ such that
    \begin{equation}\label{eq:optitd}
        \sum_{i=1}^{|\bx|}c_{\omega,i}(u^*_i) = 
        \!\!\!\max_{\bu \in V(H)^{|\bx|}}\Bigl\{\,\sum_{i=1}^{|\bx|}c_{\omega,i}(u_i) \Bigm|
        H \models \omega'(\bu)\,\Bigr\},
    \end{equation}
    and then return the best value found for~$\bu^*$.
    The input $H$ to the optimization problem (\ref{eq:optitd})
    comes from a graph class with bounded tree-depth. 
    Using Courcelle's theorem~\cite{Courcelle1990}
    one can solve a wide range of problems on these graphs in fpt time. Since we
    want to solve an optimization problem we require an extension of the
    original theorem. Courcelle, Makowsky, and Rotics define
    \textit{LinEMSOL}~\cite{CourcelleMR2000} as an extension of monadic second
    order logic allowing one to search for sets of vertices with weights that
    are optimal with respect to a linear evaluation function.
    
    Before we can apply the result of \textit{LinEMSOL} to $H$ and $\omega'$,
    note that the size of the signature of $H$ and $\omega'$ and the size of
    $\omega$ do not depend only on $|\bx|$, but on the weak coloring number of
    $G$. Hence, it is unsuitable as an input in this form. Instead, we apply
    \Cref{lem:restrict_signature} to $H$ and $\omega'$ resulting in a graph
    $H^*$ and a complete conjunctive clause $\omega^*$, both with a signature
    $\rho$ where the size of $\rho$ and $\omega^*$ is bounded by a function
    depending only on $|\bx|$ and $H \models \omega'(\bu)$ iff $H^* \models
    \omega^*(\bu)$ for every $\bu \in V(G)^{|\bx|}$. Our linear evaluation
    function is $\sum_{i=1}^{|\bx|}c_{\omega^*,i}(u_i)$ and our formula
    $\omega^*(\bx)$ clearly lies in monadic second order logic.
    
    We find a solution $\bu^*$ to our optimization problem~(\ref{eq:optitd}) in
    linear fpt time with LinEMSOL. Cycling through each $\omega \in \Omega$ and
    $H \in \mathcal H$ increases the running time by a factor of $|\Omega| \cdot
    |\mathcal{H}| \leq \wcol_{2}(G)^{O(|\phi|)} \cdot
    \wcol_{2^{{|\bx|}-1}}(G)^{|\bx|}$.
    Transforming $G$ into $G'$ and $H$ into $H^*$ takes linear time.
    %
\end{proof}

\fi
\section{Hardness Results}

In this section, we try to see how far the above result can or cannot be
extended to either a bigger class of problems or to more general graph classes.
Exemplary, we examine the distance-$r$ versions of the dominating set,
independent set and clique problem. Note that in contrast to the section before,
we do not consider the partial problem versions. We see that each of these
problems behave differently in this context. The distance-$r$ dominating set
problem is already hard for distance $1$ on some almost nowhere dense graph
classes, whereas distance-$r$ independent set and distance-$r$ clique are both
fpt on almost nowhere dense graph classes.

As for graph classes, we consider classes that are closed under removing edges
because monotone graph classes are very well understood and the notions of
nowhere dense and almost nowhere dense coincide on those classes. Interestingly,
for graph classes closed under removing edges the distance-$r$ clique problem is
fpt for all distances $r$ if and only if the class is almost nowhere dense
(under some complexity theoretic assumptions). However, there exist graph
classes which are closed under removing edges but not almost nowhere dense that
allow for fpt algorithms for the distance-$r$ independent set problem. The
difference of behavior between distance-$r$ clique and distance-$r$ independent
set is also intriguing as the FO-formulation of these problems has exactly one
quantifier alternation for both. 
\iflongpaper
Before we continue, we give a formal definition
of distance-$r$ clique.

\begin{definition}
    A set $K$ of vertices is a \emph{distance-$r$ clique} in a graph $G$ if
    there exist pairwise vertex disjoint paths of length at most $r$ between
    each pair of vertices in $K$.
\end{definition}
Note that $K$ is a distance-$r$ clique if and only if $K$ are the principle
vertices of an $(r-1)$-subdivision of a clique appearing as subgraph in $G$. 
A distance-$1$ clique is exactly a \enquote{usual} clique. Note that stars are
\emph{not} distance-2 cliques. 

We also consider the generalization of an independent set.
\begin{definition}
    A set $I$ of vertices is a \emph{distance-$r$ independent set} in a graph
    $G$ if every distinct pair of vertices from $I$ has distance strictly larger
    than $r$ in $G$.
\end{definition}
Note that a usual independent set is exactly a distance-$1$ independent set.

\fi

\subsection{Exact Evaluation Beyond Nowhere Dense Classes}

The following lemma proves that dominating set is W[1]-hard on almost nowhere
dense classes.

The class of bipartite graphs with sides $L$ and $R$ where $L$ has
polylogarithmic size is almost nowhere dense: A witness for this is a vertex
ordering that starts with $L$ and starts . Only the vertices from $L$ are weakly
$r$-reachable from any vertex. Hence, $\wcol_r(G) \leq |L| +1$ for each $r$.

\begin{theorem}\label{thm:hardness-domset} In bipartite graphs whose left side
    has $2k(k-1)\lceil\log(n)\rceil$ vertices and whose right side has $n$
    vertices it is W[1]-hard to decide whether there are $k \choose 2$
    right-side vertices dominating all left-side vertices.
\end{theorem}

\begin{proof}
    We reduce from colorful clique. Assume we have a $k$-partite graph $G$ of
    size $n$ consisting of parts $V_0,\dots,V_{k-1}$ (each of a different color)
    and want to find a colorful clique of size $k$. Without loss of generality,
    we can assume $n$ to be large enough that ${2\lceil\log(n)\rceil \choose
    \lceil\log(n)\rceil-1} \ge n$. This means, we can find for each $v \in V(G)$
    a unique binary encoding $\enc(v)$ of length $2\lceil \log(n)\rceil$ such
    that the first bit is set to one and in total exactly half the bits are set
    to one. Let $\enccomp(v)$ be the binary complement of $\enc(v)$. We
    construct a bipartite graph $H$, whose left side is partitioned into cells
    $C_{ij}$ for $0 \le i\neq j < k$, each of size $2\lceil \log(n) \rceil$, and
    whose right side will be specified soon. The vertices of each cell are
    ordered. When we say for a given vertex $v$ from the right side and cell $C$
    that \emph{$v$ is connected to $C$ according to a specified encoding}, we
    mean that for $1 \le l \le 2\lceil \log(n) \rceil$, $v$ is connected to the
    $l$th vertex of $C$ if and only if the $l$th bit in the encoding is set to
    one. For $0 \le i < k$ we define
    $$
    \suc_i(j) = 
    \begin{cases}
    j+1 \text{~mod~} k & i\neq j+1 \text{~mod~} k \\
    j+2 \text{~mod~} k & \, \text{otherwise}.
    \end{cases}
    $$
    For all $0 \le i < j < k$ and all $u \in V_i$ and $v \in V_{j}$ such that
    $uv \in E(G)$, add a vertex $x_{u,v}$ to the right side and
    \begin{itemize}
        \item connect $x_{u,v}$ to $C_{i,j}$ according to $\enc(u)$,
        \item connect $x_{u,v}$ to $C_{i,\suc_i(j)}$ according to $\enccomp(u)$,
        \item connect $x_{u,v}$ to $C_{j,i}$ according to $\enc(v)$,
        \item connect $x_{u,v}$ to $C_{j,\suc_j(i)}$ according to $\enccomp(v)$.
    \end{itemize}
\iflongpaper
    \textbf{Correctness:}
    We claim the correctness of our construction: $G$ contains a colorful clique
    of size $k$ if and only if $H$ contains a set of at most $k \choose 2$
    right-side vertices dominating all left-side vertices.
    
    The forward direction is easy. If $G$ contains a colorful clique
    $v_0,\dots,v_{k-1}$ then it is easy to see that the set $\{ x_{v_i,v_j} \mid
    0 \le i < j < k\}$ dominates all left-side vertices.
    
    For the backward direction, assume there exists a set $S$ of at most $k
    \choose 2$ right-side vertices that dominates all left-side vertices. We say
    a vertex \emph{touches} a cell if it is adjacent to at least one vertex from
    the cell. There are $k(k-1)$ cells, each right-side vertex touches most four
    cells, and each cell needs to be touched by at least two vertices from $S$.
    Thus, with $|S| \le k(k-1)/2$ and by a simple counting argument, $S$ can
    only dominate all left-side vertices is if all cells are touched by exactly
    two vertices from $S$.
    
    Let us fix a cell $C$. There exist exactly two vertices $x,y \in S$ touching
    $C$. Both $x$ and $y$ are adjacent to exactly half the vertices of $C$,
    meaning that every vertex in $C$ has exactly one neighbor from $x$ and $y$.
    For every vertex $v \in V(G)$, the encoding $\enc(v)$ has the first bit set
    to one. Thus, there exists a vertex $v \in V(G)$ such that one vertex from
    $x$ and $y$ is connected to $C$ according to $\enc(v)$ and the other vertex
    from $x$ and $y$ is connected to $C$ according to $\enccomp(v)$.
    
    Assume a cell $C_{i,j}$ is connected to a vertex $x \in S$ according to
    $\enc(v)$ for some vertex $v \in V(G)$. By the way the adjacency of $x$ is
    defined, $C_{i,\suc_i(j)}$ is connected to $x$ according to $\enccomp(v)$.
    By the previous paragraph, there exists a vertex from $S$ such that
    $C_{i,\suc_i(j)}$ is connected to this vertex according to $\enc(v)$. By
    induction, for all $0 \le i < k$, there exists a vertex $v_i$ such that for
    all $0 \le j \neq i < k$, each cell $C_{i,j}$ is connected to some vertex
    from $S$ according to $\enc(v_i)$.
    
    For all $0 \le i < j < k$, the vertex that touches $C_{i,j}$ according to
    $\enc(v_i)$ also touches $C_{j,i}$ according to $\enc(v_j)$. This guarantees
    that there is an edge between $v_i$ and $v_j$ in $G$. Therefore, the
    vertices $v_0,\dots,v_{k-1}$ form a clique of size $k$ in $G$.
\fi
\end{proof}

We can reduce the aforementioned dominating set variation to the classical
dominating set problem by connecting the right side to a fresh vertex.
\begin{corollary} \label{cor:pds-hard-awnd}
    There exists an almost nowhere dense graph class $\cal C$ where the dominating set
    problem is $W[1]$-hard and cannot be solved in time $n^{o(k)}$ assuming ETH. This implies also the hardness of the fragments PDS-like, \FOCONE, and \FOCless of \FOC on $\cal C$.
\end{corollary}
Note that this result does not follow from the intractability result of FO-logic on subgraph-closed somewhere dense classes, i.e. not nowhere dense classes.

\subsection{Beyond Distance One}

We showed that the dominating set problem is \Ww-hard on some almost nowhere
dense graph class. However, this is not true for the distance-$r$ clique and
independent set problem.

Distance-$r$ clique and independent set on the other hand are fpt on almost
nowhere dense graph classes. Here, we use low treedepth colorings to solve
existential FO formulas. With the right formulation and inclusion-exclusion this
works even for distance-$r$ independent set which cannot be expressed as a
purely existential FO formula.
\begin{theorem}\label{thm:hardness-clique} There exists a computable function
    $f$ such that for every graph $G$ the distance-$r$ clique problem can be
    solved in time $\wcol_{f(k,r)}(G)^{f(k,r)} n$.
\end{theorem}
\begin{proof}
    We can solve this problem with the help of subgraph queries where each
    subgraph is an ${\le}r$-subdivision of a $k$-clique. These subgraphs have
    less than $k^2(r+1)$ vertices and there are at most $(r+1)^{k^2}$ of them.
    Subgraph queries can be done by checking an existential FO-formula using
    \Cref{thm:opti}.
\end{proof}

\begin{theorem}\label{thm:indset-algo} There exists a computable function $f$
    such that for each graph $G$ the distance-$r$ independent set problem can be
    solved in time $\wcol_{f(k,r)}(G)^{f(k,r)}n$.
\end{theorem}
%

\iflongpaper
\begin{proof}
    A \emph{distance-$r$ $k$-subrelation} is a function 
    $D \colon \binom{[k]}{2} \to [r] \cup \{\infty, *\}$.
    We write $(G, h)\models D$ for a graph $G$ and an injective function $h \colon [k] \to V(G)$ 
    if for  every $vw \in \binom{[k]}{2}$,
    \begin{enumerate}
        \item if $D(vw) = l \in [r]$, then $\dist_G(h(v), h(w)) \leq l$,
        \item if $D(vw) = \infty$, then $\dist_G(h(v), h(w)) \geq r+1$.
    \end{enumerate}
    Let $[[D]]^G$ be the number of functions $h$ with $(G, h) \models D$.
    In essence, $D$ encodes whether some graph $H$ appears in $G$ as a subgraph with conditions on the distances of non-adjacent vertices. 
    
    Let $D^\infty$ be the distance-$r$ $k$-subrelation with $D^\infty(vw) = \infty$ for every $vw$.
    There exists a distance-$r$ independent set in a graph  $G$ if and only if 
    there is some $h$ such that $(G,h) \models D^\infty$, 
    in particular, if and only if $[[D^\infty]]^G > 0$.

    To compute this value, we use the inclusion-exclusion principle.
    Let $D$ be some distance-$r$ $k$-subrelation with an entry $vw$ such that $D(vw) = \infty$.
    Then the value of $[[D]]^G$ can be computed as $[[D]]^G = [[D^*]]^G - [[D^{r}]]^G$ 
    where $D^*$ and $D^{r}$ are distance-$r$ $k$-subrelations equal to $D$ 
    except for the value of $vw$ which is $*$ and $r$ respectively.
    We apply this rule exhaustively until $\infty$ does not appear in the images of the subrelations.

    Distance-$r$ $k$-subrelations without $\infty$ can be expressed by a disjunction of subgraph queries
    where each graph is an $r$-subdivision of graph on $k$ vertices.
    These graphs have less than $k^2(r+1)$ vertices and there are at most $(r+2)^{k^2}$ such graphs.
    Using \Cref{lem:wcol-approx} and \Cref{prop:treedepth-col} we can compute a $k^2(r+1)$-treedepth coloring
    with $\wcol_{f(k^2(r+1))}(G)^{f(k^2(r+1))}$ many colors.
    Using these low treedepth colorings and \cite[Theorem 6 and 8]{DemaineRRVSS19} 
    one can count how often such graphs appear as subgraphs in time $\wcol_{f'(k, r)}(G)^{f'(k, r)} n$ for some computable function $f'$.
\end{proof}
\fi

\subsection{Beyond Almost Nowhere Dense}

For graph classes that are closed under removing \emph{vertices and edges},
i.e., monotone graph classes, we know a lot already. Most importantly, FO-model
checking is fpt on such classes if and only if the class is nowhere dense
(unless \FPT = \Ww) \cite{GroheKS17}. We now want to consider graph classes that
are only closed under removing \emph{edges}. Here the concept of almost nowhere
dense graph classes becomes interesting.


The following observation follows directly from characterization
\ref{col_anwd_char} in \Cref{thm:characterization_almost_nowhere_dense}. If
$\cal P$ is a parameterized problem that can be solved in time
$\col_{f(k)}(G)^{f(k)}n$ and $\cal C$ is an almost nowhere dense graph class,
then $\cal P$ can be solved on $\cal C$ in almost linear fpt time $f(k,
\varepsilon)n^{1+\varepsilon}$ for every $\varepsilon > 0$. We complement this
by showing that the distance-$r$ clique problem is most likely not fpt on all
graph classes that are not almost nowhere dense, but closed under removing
edges. Hence, under certain common complexity theoretic assumptions, if a graph
class $\cal C$ is closed under removal of edges then distance-$r$ clique is fpt
on $\cal C$ iff $\cal C$ is almost nowhere dense.

\begin{theorem}\label{thm:hardness} 
    Let $\cal C$ be a graph class that is not
    almost nowhere dense, but closed under removing edges. Then there exists a
    number $r$, such that one cannot solve the distance-$r'$ clique problem
    parameterized by solution size in fpt time on $\cal C$ for all $r' \le r$
    unless {\rm i.o.W[1] $\subseteq$ FPT}.
\end{theorem}

Similar hardness results in parameterized complexity are usually built on the
hardness assumption FPT $\neq$ \Ww{}. The complexity class $i.o.$W[1] should be
read as ``infinitely often in W[1]'' and needs to be explained.
\begin{definition}
For a language $L$ and an integer $n$ let $L_n= L \cap \{0,1\}^n$. A language
$L$ is in $i.o.C$ for a complexity class $C$ if there is some $L' \in C$ such
that $L'_n = L_n$ for infinitely many input lengths $n$.
\end{definition}
Considering the infinite often variant $i.o.C$ of a complexity class $C$ is an
established technique in complexity theory (i.e.,
\cite{4567890,beigel2006infinitely}).
To prove our result, we show that a graph class $\cal C$ that is
not almost nowhere dense, contains an infinite sequence of graphs
having
cliques of polynomial size as bounded depth topological minors. If $\cal C$ is
also closed under removal of edges then having bounded depth
topological clique minors of size $n$ implies the existence of subdivisions of
arbitrary graphs $H$ of size $n$ as induced subgraphs. Extra care needs to be
taken to make sure that all paths connecting the principal vertices should be of
equal length, since otherwise a reduction would need to try out an exponential
number of possible length combinations to finally find the correct subdivision
of $H$ that is contained in $\cal C$. The following corollary is a direct
consequence of
\Cref{thm:characterization_almost_nowhere_dense}.\ref{subdiv_anwd_char}.

\begin{corollary}\label{lem:largecliques} Let $\cal C$ be some graph class that
    is not almost nowhere dense. Then there are $r$, $\varepsilon$ and an
    infinite sequence of strictly ascending numbers $n_0,n_1,\dots$ such that
    for all $i \in \N$ there is a graph $G \in \cal C$ of order at most $n_i$
    that contains an $r'$-subdivision of $K_{\lceil n_i^\varepsilon \rceil}$ as
    a subgraph for some $r' \le r$.
\end{corollary}

The consequence $\rm i.o.W[1]\subseteq FTP$ is weaker than $\rm W[1]\subseteq FPT$.  We could use
the latter in Theorem~\ref{thm:hardness} if we required a stronger
precondition, i.e., that $\C$ has ``witnesses'' for input lengths
$n_0,n_1,n_2,\ldots$ such that the gap between $n_i$ and $n_{i+1}$ is
only polynomial.  This approach has been used, e.g., in proving lower
bounds on the running time of MSO-model checking in graph classes
where the treewidth grows polylogarithmically~\cite{MSO-hard,ganian2014lower}.


\begin{proof}[Proof of \Cref{thm:hardness}]
    Let $r$ and $\varepsilon$ be the constants (depending on $\mathcal{C}$) from
    \Cref{lem:largecliques}. Assume that the distance-$(r+1)$ clique problem on
    $\mathcal{C}$ is fpt when parameterized by solution size. We will present a
    Turing reduction showing that the (usual) clique problem on the class of all
    graphs is infinitely often in FPT.

    By \Cref{lem:largecliques} for infinitely many $n_0,n_1,\dots \in \N$ there
    exists a graph from $\cal C$ of size at most $n_i^{1/\varepsilon}$ that
    contains an $r'$-subdivision of a clique of size $n_i$ as a subgraph for
    some $r' \le r$. Let us pick one $n=n_i$. Suppose we want to decide whether
    a graph~$G$ with $n$ vertices contains a clique of size $k$. Since $\cal C$
    is closed under removal of edges, there exist $r' \le r$, and $n \le N \le
    n^{1/\varepsilon}$ such that $\cal C$ contains a graph $H_{r',N}$ consisting
    of an $r'$-subdivision of $G$ together with $N$ isolated vertices. Now for
    all $k$, $G$ contains a clique of size $k$ iff $H_{r',N}$ contains a
    distance-$(r'+1)$ clique of size $k$. Assume for contradiction we had an
    algorithm that decides in time at most $f(r',k)n^c$ whether a graph in $\cal
    C$ of size $n$ contains an distance-$(r'+1)$ clique for $r' \le r$. (For
    graphs not in $\cal C$, the algorithm may give a
    wrong answer, but we can modify it to construct and test a
    witness of a distance-$(r'+1)$ clique on yes-instances. Hence, we can assume
    that the algorithm never returns ``no'' on yes-instances.)

    The existence of such an algorithm yields us an FPT algorithm for the
    $k$-clique problem on general graphs: For all $r' \le r$, and $n \le N \le
    n^{1/\varepsilon}$, we run this (hypothetical) fpt algorithm in parallel on
    $H_{r',N}$ for $f(r',k)N^c$ time steps. Then $G$ contains a clique of size
    $k$ iff for at least one value of $r'$ and $N$ we have $H_{r',N} \in \cal C$
    and $H_{r,N}$ contains a distance-$(r'+1)$ $k$-clique.
    
    As the $k$-clique problem is $W[1]$-hard, we get the desired result.
\end{proof}

Note that this result does not extend to the distance-$r$ independent
set problem.  Consider the class of graphs where at least half
of its vertices are isolated.  Then the distance-$r$ independent set
problem is trivially FPT for this graph class.  However, this graph
class is closed under removing edges, but it is not almost nowhere dense.

\bibliographystyle{plainurl}
\bibliography{../wcol}

\begin{thebibliography}{10}

\bibitem{AminiFS2011}
Omid Amini, Fedor~V. Fomin, and Saket Saurabh.
\newblock Implicit branching and parameterized partial cover problems.
\newblock {\em J. Comput. Syst. Sci.}, 77(6):1159--1171, 2011.
\newblock \href {https://doi.org/10.1016/j.jcss.2010.12.002}
  {\path{doi:10.1016/j.jcss.2010.12.002}}.

\bibitem{beigel2006infinitely}
Richard Beigel, Lance Fortnow, and Frank Stephan.
\newblock Infinitely-often autoreducible sets.
\newblock {\em SIAM Journal on Computing}, 36(3):595--608, 2006.

\bibitem{4567890}
Leonard Berman.
\newblock On the structure of complete sets: Almost everywhere complexity and
  infinitely often speedup.
\newblock In {\em 17th Annual Symposium on Foundations of Computer Science
  (sfcs 1976)}, pages 76--80, 1976.
\newblock \href {https://doi.org/10.1109/SFCS.1976.22}
  {\path{doi:10.1109/SFCS.1976.22}}.

\bibitem{Courcelle1990}
Bruno Courcelle.
\newblock The monadic second-order logic of graphs. {I}. {R}ecognizable sets of
  finite graphs.
\newblock {\em Inf. Comput.}, 85(1):12--75, 1990.
\newblock \href {https://doi.org/10.1016/0890-5401(90)90043-H}
  {\path{doi:10.1016/0890-5401(90)90043-H}}.

\bibitem{CourcelleMR2000}
Bruno Courcelle, Johann~A. Makowsky, and Udi Rotics.
\newblock Linear time solvable optimization problems on graphs of bounded
  clique-width.
\newblock {\em Theory Comput. Syst.}, 33(2):125--150, 2000.
\newblock \href {https://doi.org/10.1007/s002249910009}
  {\path{doi:10.1007/s002249910009}}.

\bibitem{DawarGK2007}
Anuj Dawar, Martin Grohe, and Stephan Kreutzer.
\newblock Locally excluding a minor.
\newblock In {\em 22nd {IEEE} Symposium on Logic in Computer Science {(LICS}
  2007), 10-12 July 2007, Wroclaw, Poland, Proceedings}, pages 270--279. {IEEE}
  Computer Society, 2007.
\newblock \href {https://doi.org/10.1109/LICS.2007.31}
  {\path{doi:10.1109/LICS.2007.31}}.

\bibitem{DawarK2009}
Anuj Dawar and Stephan Kreutzer.
\newblock Domination problems in nowhere-dense classes.
\newblock In Ravi Kannan and K.~Narayan Kumar, editors, {\em {IARCS} Annual
  Conference on Foundations of Software Technology and Theoretical Computer
  Science, {FSTTCS} 2009, December 15-17, 2009, {IIT} Kanpur, India}, volume~4
  of {\em LIPIcs}, pages 157--168. Schloss Dagstuhl - Leibniz-Zentrum f{\"{u}}r
  Informatik, 2009.
\newblock \href {https://doi.org/10.4230/LIPIcs.FSTTCS.2009.2315}
  {\path{doi:10.4230/LIPIcs.FSTTCS.2009.2315}}.

\bibitem{DemaineRRVSS19}
Erik~D. Demaine, Felix Reidl, Peter Rossmanith, Fernando~S{\'{a}}nchez
  Villaamil, Somnath Sikdar, and Blair~D. Sullivan.
\newblock Structural sparsity of complex networks: Bounded expansion in random
  models and real-world graphs.
\newblock {\em J. Comput. Syst. Sci.}, 105:199--241, 2019.
\newblock \href {https://doi.org/10.1016/j.jcss.2019.05.004}
  {\path{doi:10.1016/j.jcss.2019.05.004}}.

\bibitem{DowneyF1999}
Rodney~G. Downey and Michael~R. Fellows.
\newblock {\em Parameterized Complexity}.
\newblock Monographs in Computer Science. Springer, 1999.
\newblock \href {https://doi.org/10.1007/978-1-4612-0515-9}
  {\path{doi:10.1007/978-1-4612-0515-9}}.

\bibitem{DowneyF2013}
Rodney~G. Downey and Michael~R. Fellows.
\newblock {\em Fundamentals of Parameterized Complexity}.
\newblock Texts in Computer Science. Springer, 2013.
\newblock \href {https://doi.org/10.1007/978-1-4471-5559-1}
  {\path{doi:10.1007/978-1-4471-5559-1}}.

\bibitem{DreierR2021}
Jan Dreier and Peter Rossmanith.
\newblock Approximate evaluation of first-order counting queries.
\newblock In D{\'{a}}niel Marx, editor, {\em Proceedings of the 2021 {ACM-SIAM}
  Symposium on Discrete Algorithms, {SODA} 2021, Virtual Conference, January 10
  - 13, 2021}, pages 1720--1739. {SIAM}, 2021.
\newblock \href {https://doi.org/10.1137/1.9781611976465.104}
  {\path{doi:10.1137/1.9781611976465.104}}.

\bibitem{DurandG2007}
Arnaud Durand and Etienne Grandjean.
\newblock First-order queries on structures of bounded degree are computable
  with constant delay.
\newblock {\em {ACM} Trans. Comput. Log.}, 8(4):21, 2007.
\newblock \href {https://doi.org/10.1145/1276920.1276923}
  {\path{doi:10.1145/1276920.1276923}}.

\bibitem{DurandSS22}
Arnaud Durand, Nicole Schweikardt, and Luc Segoufin.
\newblock Enumerating answers to first-order queries over databases of low
  degree.
\newblock {\em Log. Methods Comput. Sci.}, 18(2), 2022.
\newblock \href {https://doi.org/10.46298/lmcs-18(2:7)2022}
  {\path{doi:10.46298/lmcs-18(2:7)2022}}.

\bibitem{DvorakThesis}
Z.~\Dvorak.
\newblock {\em Asymptotical Structure of Combinatorial Objects}.
\newblock PhD thesis, Charles University, Faculty of Mathematics and Physics,
  2007.

\bibitem{DvorakKT2013}
Zdenek Dvor{\'{a}}k, Daniel Kr{\'{a}}l$\!$', and Robin Thomas.
\newblock Testing first-order properties for subclasses of sparse graphs.
\newblock {\em J. {ACM}}, 60(5):36:1--36:24, 2013.
\newblock \href {https://doi.org/10.1145/2499483} {\path{doi:10.1145/2499483}}.

\bibitem{FrickG2001}
Markus Frick and Martin Grohe.
\newblock Deciding first-order properties of locally tree-decomposable
  structures.
\newblock {\em J. {ACM}}, 48(6):1184--1206, 2001.
\newblock \href {https://doi.org/10.1145/504794.504798}
  {\path{doi:10.1145/504794.504798}}.

\bibitem{ganian2014lower}
Robert Ganian, Petr Hlin{\v{e}}n{\'y}, Alexander Langer, Jan
  Obdr{\v{z}}{\'a}lek, Peter Rossmanith, and Somnath Sikdar.
\newblock Lower bounds on the complexity of {MSO1} model-checking.
\newblock {\em Journal of Computer and System Sciences}, 80(1):180--194, 2014.

\bibitem{GolovachV2008}
Petr~A. Golovach and Yngve Villanger.
\newblock Parameterized complexity for domination problems on degenerate
  graphs.
\newblock In Hajo Broersma, Thomas Erlebach, Tom Friedetzky, and Dani{\"{e}}l
  Paulusma, editors, {\em Graph-Theoretic Concepts in Computer Science, 34th
  International Workshop, {WG} 2008, Durham, UK, June 30 - July 2, 2008.
  Revised Papers}, volume 5344 of {\em Lecture Notes in Computer Science},
  pages 195--205, 2008.
\newblock \href {https://doi.org/10.1007/978-3-540-92248-3\_18}
  {\path{doi:10.1007/978-3-540-92248-3\_18}}.

\bibitem{Grohe01}
Martin Grohe.
\newblock Generalized model-checking problems for first-order logic.
\newblock In Afonso Ferreira and Horst Reichel, editors, {\em {STACS} 2001,
  18th Annual Symposium on Theoretical Aspects of Computer Science, Dresden,
  Germany, February 15-17, 2001, Proceedings}, volume 2010 of {\em Lecture
  Notes in Computer Science}, pages 12--26. Springer, 2001.
\newblock \href {https://doi.org/10.1007/3-540-44693-1\_2}
  {\path{doi:10.1007/3-540-44693-1\_2}}.

\bibitem{GroheKRSS18}
Martin Grohe, Stephan Kreutzer, Roman Rabinovich, Sebastian Siebertz, and
  Konstantinos~S. Stavropoulos.
\newblock Coloring and covering nowhere dense graphs.
\newblock {\em {SIAM} J. Discret. Math.}, 32(4):2467--2481, 2018.
\newblock \href {https://doi.org/10.1137/18M1168753}
  {\path{doi:10.1137/18M1168753}}.

\bibitem{GroheKS2013}
Martin Grohe, Stephan Kreutzer, and Sebastian Siebertz.
\newblock Characterisations of nowhere dense graphs (invited talk).
\newblock In Anil Seth and Nisheeth~K. Vishnoi, editors, {\em {IARCS} Annual
  Conference on Foundations of Software Technology and Theoretical Computer
  Science, {FSTTCS} 2013, December 12-14, 2013, Guwahati, India}, volume~24 of
  {\em LIPIcs}, pages 21--40. Schloss Dagstuhl - Leibniz-Zentrum f{\"{u}}r
  Informatik, 2013.
\newblock \href {https://doi.org/10.4230/LIPIcs.FSTTCS.2013.21}
  {\path{doi:10.4230/LIPIcs.FSTTCS.2013.21}}.

\bibitem{GroheKS17}
Martin Grohe, Stephan Kreutzer, and Sebastian Siebertz.
\newblock Deciding first-order properties of nowhere dense graphs.
\newblock {\em J. {ACM}}, 64(3):17:1--17:32, 2017.
\newblock \href {https://doi.org/10.1145/3051095} {\path{doi:10.1145/3051095}}.

\bibitem{GroheS18}
Martin Grohe and Nicole Schweikardt.
\newblock First-order query evaluation with cardinality conditions.
\newblock In Jan~Van den Bussche and Marcelo Arenas, editors, {\em Proceedings
  of the 37th {ACM} {SIGMOD-SIGACT-SIGAI} Symposium on Principles of Database
  Systems, Houston, TX, USA, June 10-15, 2018}, pages 253--266. {ACM}, 2018.
\newblock \href {https://doi.org/10.1145/3196959.3196970}
  {\path{doi:10.1145/3196959.3196970}}.

\bibitem{KazanaS2013}
Wojciech Kazana and Luc Segoufin.
\newblock First-order queries on classes of structures with bounded expansion.
\newblock {\em Log. Methods Comput. Sci.}, 16(1), 2020.
\newblock \href {https://doi.org/10.23638/LMCS-16(1:25)2020}
  {\path{doi:10.23638/LMCS-16(1:25)2020}}.

\bibitem{coloringdefinition}
Henry~A. Kierstead and Daqing Yang.
\newblock Orderings on graphs and game coloring number.
\newblock {\em Order}, 20(3):255--264, 2003.
\newblock \href {https://doi.org/10.1023/B:ORDE.0000026489.93166.cb}
  {\path{doi:10.1023/B:ORDE.0000026489.93166.cb}}.

\bibitem{KneisMR2007}
Joachim Kneis, Daniel M{\"{o}}lle, and Peter Rossmanith.
\newblock Partial vs. complete domination: $t$-dominating set.
\newblock In Jan van Leeuwen, Giuseppe~F. Italiano, Wiebe van~der Hoek,
  Christoph Meinel, Harald Sack, and Frantisek Plasil, editors, {\em {SOFSEM}
  2007: Theory and Practice of Computer Science, 33rd Conference on Current
  Trends in Theory and Practice of Computer Science, Harrachov, Czech Republic,
  January 20-26, 2007, Proceedings}, volume 4362 of {\em Lecture Notes in
  Computer Science}, pages 367--376. Springer, 2007.
\newblock \href {https://doi.org/10.1007/978-3-540-69507-3\_31}
  {\path{doi:10.1007/978-3-540-69507-3\_31}}.

\bibitem{MSO-hard}
Stephan Kreutzer and Siamak Tazari.
\newblock Lower bounds for the complexity of monadic second-order logic.
\newblock In {\em Proceedings of the 25th Annual {IEEE} Symposium on Logic in
  Computer Science, {LICS} 2010, 11-14 July 2010, Edinburgh, United Kingdom},
  pages 189--198. {IEEE} Computer Society, 2010.
\newblock \href {https://doi.org/10.1109/LICS.2010.39}
  {\path{doi:10.1109/LICS.2010.39}}.

\bibitem{KuskeS2017}
Dietrich Kuske and Nicole Schweikardt.
\newblock First-order logic with counting.
\newblock In {\em 32nd Annual {ACM/IEEE} Symposium on Logic in Computer
  Science, {LICS} 2017, Reykjavik, Iceland, June 20-23, 2017}, pages 1--12.
  {IEEE} Computer Society, 2017.
\newblock \href {https://doi.org/10.1109/LICS.2017.8005133}
  {\path{doi:10.1109/LICS.2017.8005133}}.

\bibitem{NesetrilM2008a}
Jaroslav Nešetřil and Patrice~Ossona de~Mendez.
\newblock Grad and classes with bounded expansion {I}. decompositions.
\newblock {\em Eur. J. Comb.}, 29(3):760--776, 2008.
\newblock \href {https://doi.org/10.1016/j.ejc.2006.07.013}
  {\path{doi:10.1016/j.ejc.2006.07.013}}.

\bibitem{sparsity}
Jaroslav Nešetřil and Patrice~Ossona de~Mendez.
\newblock {\em Sparsity - Graphs, Structures, and Algorithms}, volume~28 of
  {\em Algorithms and combinatorics}.
\newblock Springer, 2012.
\newblock \href {https://doi.org/10.1007/978-3-642-27875-4}
  {\path{doi:10.1007/978-3-642-27875-4}}.

\bibitem{nesetril2011600}
Jaroslav Nešetřil and Patrice {Ossona de Mendez}.
\newblock On nowhere dense graphs.
\newblock {\em European Journal of Combinatorics}, 32(4):600 -- 617, 2011.
\newblock URL:
  \url{http://www.sciencedirect.com/science/article/pii/S0195669811000151},
  \href {https://doi.org/https://doi.org/10.1016/j.ejc.2011.01.006}
  {\path{doi:https://doi.org/10.1016/j.ejc.2011.01.006}}.

\bibitem{notes}
Marcin Pilipczuk, Micha\l{} Pilipczuk, and Sebastian Siebertz.
\newblock Lecture notes for the course ``{S}parsity'' given at {F}aculty of
  {M}athematics, {I}nformatics, and {M}echanics of the {U}niversity of
  {W}arsaw, Winter semesters 2017/18 and 2019/20.
\newblock Available \url{https://www.mimuw.edu.pl/~mp248287/sparsity2}.

\bibitem{Seese1996}
Detlef Seese.
\newblock Linear time computable problems and first-order descriptions.
\newblock {\em Math. Struct. Comput. Sci.}, 6(6):505--526, 1996.
\newblock \href {https://doi.org/10.1017/s0960129500070079}
  {\path{doi:10.1017/s0960129500070079}}.

\bibitem{Siebertz2016}
Sebastian Siebertz.
\newblock {\em Nowhere Dense Classes of Graphs: Characterisations and
  Algorithmic Meta-Theorems}.
\newblock PhD thesis, TU Berlin, 2016.

\bibitem{vigny}
Alexandre Vigny.
\newblock Dynamic query evaluation over structures with low degree.
\newblock {\em CoRR}, abs/2010.02982, 2020.
\newblock URL: \url{https://arxiv.org/abs/2010.02982}, \href
  {http://arxiv.org/abs/2010.02982} {\path{arXiv:2010.02982}}.

\bibitem{Zhu2009}
Xuding Zhu.
\newblock Colouring graphs with bounded generalized colouring number.
\newblock {\em Discret. Math.}, 309(18):5562--5568, 2009.
\newblock \href {https://doi.org/10.1016/j.disc.2008.03.024}
  {\path{doi:10.1016/j.disc.2008.03.024}}.

\end{thebibliography}
\end{document}